\documentclass[runningheads]{elsarticle}

\usepackage{amsmath,amssymb,amsfonts,stmaryrd}
\usepackage{algorithm}
\usepackage[noend]{algpseudocode}

\usepackage{graphicx}
\usepackage[usenames,dvipsnames]{color}
\usepackage{hyperref}
\usepackage[utf8]{inputenc}
\usepackage[T1]{fontenc}
\usepackage{booktabs}
\usepackage{xspace}
\usepackage{xcolor}
\usepackage{array}
\usepackage{listings}
\usepackage{siunitx} 

\usepackage{comment}

\newif\ifcomments%
\commentstrue%

\newcommand{\R}{\ensuremath{\mathbb{R}}}
\newcommand{\C}{\ensuremath{\mathbb{C}}}

\newcommand{\D}{{\ensuremath{\cal D}}}
\newcommand{\F}{{\ensuremath{\cal F}}}

\newcommand{\diff}[1]{{#1}}
\newcommand{\sout}[1]{}

\def\url#1{\texttt{#1}}
\newtheorem{definition}{Definition}

\newtheorem{theorem}[definition]{Theorem}
\newtheorem{lemma}[definition]{Lemma}
\newtheorem{proposition}[definition]{Proposition}

\newtheorem{example}[definition]{Example}
\newenvironment{proof}{\par\noindent\normalfont\emph{Proof.}}{\hfill\qed\bigskip}

\journal{Theoretical Computer Science C natural computing}
\begin{document}
\begin{frontmatter}

\title{On a model of online analog computation in the cell \diff{with absolute functional robustness}: algebraic characterization, \diff{function} compiler and error control}
  
\author[l1]{Mathieu Hemery}
\author[l1]{Fran\c{c}ois Fages}
\address[l1]{Inria Saclay, EPI Lifeware, Palaiseau, France}


\begin{abstract}
 The Turing completeness of continuous Chemical Reaction Networks (CRNs) states that any
	computable real function can be computed by a continuous CRN on a finite set of
	molecular species, possibly restricted to elementary reactions, i.e.~with at most two
	reactants and mass action law kinetics. In this paper, we introduce a \diff{\sout{different} more stringent} notion
	of robust online analog computation\diff{, and Absolute Functional Robustness (AFR),} for the CRNs \diff{that stabilize the concentration values of some output species
          to the result of one function of the input species concentrations,
          in a perfectly robust manner with respect
	to perturbations of both intermediate and output species.} We prove that the set of real
	functions stabilized by a CRN with mass action law kinetics is precisely the set of
	real algebraic functions.  Based on this result, we present a compiler which takes as
	input any algebraic function \diff{(defined by one polynomial and one point
	for selecting one branch of the algebraic curve defined by the polynomial)}
	and generates an abstract CRN to stabilize it.  Furthermore, we provide error bounds to
	estimate and control the error of \diff{\sout{a noiseless} an unperturbed} system\diff{\sout{when}, under the
    assumption that} the \diff{environment} inputs \diff{\sout{move
	according to}are driven by} $k$-Lipschitz functions.
\end{abstract}


\end{frontmatter}
\section{Introduction}

\subsection{Background}

Chemical Reaction Networks (CRNs) are a standard formalism used in chemistry and
biology to model complex molecular interaction systems at various levels of abstraction.
In the perspective of systems biology, they are a central tool to analyze
the high-level functions of the cell in terms of their low-level molecular interactions.
In that perspective, the Systems Biology Markup Language (SBML)~\cite{Hucka03bi}
is a common format to exchange CRN models and build CRN model repositories, such as Biomodels.net~\cite{CLN13issb}
which contains thousands of CRN models of a large variety of cell biochemical processes.
In the perspective of synthetic biology, they constitute a target
programming language to implement in chemistry new functions either \textit{in vitro},
e.g.~using DNA polymers~\cite{QSW11dna}, or in living
cells using plasmids~\cite{DWGLERBW14nar} or in artificial vesicles using proteins~\cite{CAFRM18msb}.

The mathematical theory of CRNs was introduced in the late 70's, on the one hand, by Feinberg in~\cite{Feinberg77crt},
by focusing on \diff{robust} perfect adaptation \diff{(RPA)} properties\diff{, Absolute Concentration Robustness
(ACR)~\cite{SF10science}} and multi-stability analyses~\cite{CF06siamjam};
and on the other hand, by {\'E}rdi and T{\'o}th
by characterizing the set of Polynomial Ordinary Differential Equation systems (PODEs)
that can be defined by CRNs with mass action law kinetics,
using dual-rail encoding for negative variables~\cite{HT79cmsjb,ET89book,OK11iet,FLBP17cmsb}.

More recently, a computational theory of CRNs was investigated by formally relating their
Boolean, discrete, stochastic and differential
semantics in the framework of abstract interpretation~\cite{FS08tcs},
and by studying the computational power of CRNs under those different interpretations~\cite{CSWB09ab,CDS12nc,FLBP17cmsb}.

In particular, under the continuous semantics of CRNs interpreted by ODEs,
the Turing-completeness result established in~\cite{FLBP17cmsb} states that
any computable real function, i.e.~computable by a Turing machine with an arbitrary precision given in input,
can be computed by a continuous CRN on a finite set of molecular species,
using elementary reactions with at most two reactants and mass action law kinetics.
This result uses the following notion of analog computation of a non-negative real function computed by a CRN,
where the result is given by the concentration of one species, $y_1$,
and the error is controlled by the concentration of one second species, $y_2$:

\begin{definition}\cite{FLBP17cmsb}
	A function $ f: \R_+ \to \R_+ $ is CRN-computable if there exist
a CRN over some molecular species $\{y_1,...,y_n\}$,
	and a polynomial \diff{$ q: \R_+^n \to \R_+$} defining their initial concentration values,
such that for all $x \in \R_+$ there exists some (necessarily unique) function $y: \R_+
	\to \R_+^n$ such that
$y(0)=q(x),\ y'(t) = p (y (t)) $ and for all $t>1$:
$$|y_1(t)-f(x)| \le y_2(t),$$
 $y_2(t) \ge 0$, $y_2(t)$ is decreasing and 
 $ \lim_{t \to \infty} y_2(t) = 0$.
\end{definition}

From the theoretical point of view of computability, the control of the error
which is explicitly represented in the above definition by the auxiliary variable $y_2$,
is necessary to
decide when the result is ready for a requested precision, and to mathematically define the function computed by a CRN, if any.

From a practical point of view however, precision is of course an irrelevant issue
since chemical reactions are stochastic in nature and
a more  important criterion than the precision of the result is
the stability of the result and its robustness with respect to molecular concentration perturbations.
With this provision to omit error control, the Turing-completeness result of continuous CRNs was immediately used in~\cite{FLBP17cmsb}
to design a compilation pipeline to implement any mathematical elementary function in abstract chemistry.
This compiler, implemented in our CRN modeling, analysis and synthesis software Biocham\footnote{\diff{\url{https://lifeware.inria.fr/biocham4/}}}~\cite{CFS06bi},
 generates a CRN over a finite set of abstract molecular species,
through several symbolic computation steps, mainly composed of polynomialization~\cite{HFS21cmsb},
quadratization~\cite{HFS20cmsb} and lazy dual-rail encoding of negative variables.
A similar approach is undertaken in the CRN$++$ system~\cite{VSK18dna}, also related to~\cite{CTT20nc}.

It is worth remarking that in the definition above, and in our implementation in Biocham,
the input is defined by the initial concentration of the input species
which can be consumed by the synthesized CRN to compute the result.
This marks a fundamental difference with many natural CRNs which perform a kind of online computation
by adapting the response to the time evolution of the input.
This is the case for instance of the ubiquitous MAPK signaling network which computes an input/output function
well approximated by a Hill function of order 4.9~\cite{HF96pnas},
while our synthesized CRNs to compute the same function consume their input and do not correctly adapt to change
of the input value during computation~\cite{HFS21cmsb}.

\subsection{Contributions}

In this paper\footnote{This paper is an extended version of a previous communication at CMSB 2022 published in~\cite{HF22cmsb}. The main new results added to the conference proceedings are the introduction of the notion of Absolute Functional Robustness (AFR) and the mathematical analysis of the approximation error done in this new computation model by stabilization. The latter is reported in a new section of the paper devoted to error control\sout{ and generalized to a large class of population dynamical models}. The computational results presented here have been obtained with a Biocham notebook available in different formats at~\url{https://lifeware.inria.fr/wiki/Main/Software\#CMSB22}.}, we introduce a \diff{new} notion of online computation for continuous CRNs,
by opposition to our previous \diff{static} notion of \sout{static} computation of the result of a function for any initially given input~\cite{FLBP17cmsb}.
\diff{We thus study open CRNs with species partitioned into input, intermediate and output species. The
input species are supposed to be driven by the environment while the intermediate and
output species are governed by the reactions of the CRN. We are interested to
understand what relation can be imposed between the input and output species.}
Our main theorem shows that the set of input/output functions stabilized online by a CRN with mass action law kinetics,
is precisely the set of real algebraic functions.

\begin{example}\label{ex:hill}
We can illustrate this result with a simple example.
Let us consider a cell that produces a
receptor, $I$, which is transformed in an active form, $A$, when bound to an external ligand
$L$, and which stays active even after unbinding:
\begin{gather}
\begin{aligned}
  L+I &\rightarrow L+A \\
  \emptyset &\leftrightarrow I \\
  A &\rightarrow \emptyset
\end{aligned}
\end{gather}
The differential semantics with mass action law of unitary rate constant is the
Polynomial Ordinary Differential Equation (PODE):
\begin{gather}
\begin{aligned}
	\frac{dI}{dt} &= 1-I-L\diff{\cdot}I \\
\frac{dA}{dt} &= L\diff{\cdot}I-A \\
\frac{dL}{dt} &= 0
\end{aligned}
\end{gather}
	\diff{where here and thereafter, $\cdot$ denotes the natural multiplication}.
At steady state, all the derivatives are null and by eliminating $I$, we immediately
obtain the polynomial equation: $L - LA - A = 0$. Thinking of this simple CRN
as a kind of signal processing with the ligand as input and the active receptor
as output, it is thus possible to find a polynomial relation between the input and the
output. In this case, this relation entirely defines the function computed by
the CRN: $$A = f(L) = \frac{L}{1+L}.$$
For a given concentration of ligand, this is
the only stable state of the system, independently of the initial concentrations
of $A$ or $I$. This is why we say that the CRN stabilizes the function $f$.
\end{example}

Such functions, for which there exists a polynomial relation between the inputs and
output, are called algebraic functions in mathematics.  We \diff{\sout{thus}} show
here that they characterize the set of input/output functions stabilized by CRNs with mass
action law kinetics.

Furthermore, our constructive proof provides a compilation algorithm to generate a
stabilizing  CRN for any
real algebraic curve, i.e. any curve defined by the zeros of some polynomial.
We present the main symbolic transformations steps of this new compiler implemented in Biocham,
and its main differences with the previous compilation scheme, in particular concerning the quadratization problem to solve~\cite{HFS20cmsb,HFS21cmsb}.

Then in the last Sec.~\ref{sec:error}, we provide error bounds to estimate and control the
error of \diff{\sout{a noiseless} an unperturbed} system when
\diff{\sout{the inputs continuously move in their stabilization domains with a bound on their time
derivative, i.e.~according to k-Lipschitz functions.}
we suppose that the rate of change of the input functions are bounded, that is, the
concentrations of the input species cannot change arbitrarily fast. Mathematically, we
will provide error bounds for $k$-Lipschitz environement input functions.}

\subsection{Related works}

The computational framework adopted in this article is a clarification and
\diff{refinement} of
the one already envisioned by the first author in~\cite{HF19jpcb}. The purpose of this
previous work was to evolve a biochemically plausible CRN capable of stabilizing the logarithm of its
input.
Since the logarithm is not an algebraic function, we now know that the result was bound to be
an approximation, which was not known at that time.

\diff{Our work may also be related to the notion of Absolute Concentration Robustness
(ACR) developed by Feinberg~\cite{SF10science}. A \diff{CRN} species is said to display
  ACR if it has the same concentration on every stable point of the system. More generally,
  we introduce here the notion of Absolute Functional Robustness (AFR) for CRNs in which 
one output species $y$ stabilizes a function $f$ of CRN inputs $X$ 
on a particular domain,
if we have $y = f(X)$ for every equilibrium point inside
that domain.}

A recent work by Fletcher~\& al. presents results very close to
ours. In~\cite{FKTLNR21ucnc}, the authors propose a Lyapunov computational framework for computing real numbers with CRNs,
and show that the set of numbers that are \diff{equilibrium} points of such systems
is exactly the set of algebraic numbers. An interesting result of this paper is the proof
that these \diff{equilibrium} points are locally exponentially stable, a result we will use
in Sec.~\ref{sec:error} on error control. While their method analyzes a CRN starting from a unique
initial condition and computing a single number, we propose here a framework where the
input \diff{species}
may freely evolve in a domain $I \subset \R^n$, hence computing not a number but a 
function: $f: I \to \R$, and show that only algebraic functions can be computed in this manner. 
Moreover, we give a CRN compiler taking as input the polynomial representation of
any algebraic function $f$, and giving as output one elementary CRN certified to stabilize $f$.

Our CRN compiler can also be compared to the CRN synthesis results of Buisman~\& al.~who present
in~\cite{BEHL09al} a method to implement any algebraic expression by an abstract
CRN\footnote{The terminology of ``algebraic functions'' used in the title
 of~\cite{BEHL09al} refers in fact to its restriction to algebraic expressions.}.
They
rely on a direct expression of the function and a compilation process that mimics the
composition of the elementary algebraic operations. We improve their results in three
directions. First, our compilation pipeline is able to generate stabilizing CRNs for any
algebraic function, including those algebraic functions that cannot be defined by algebraic expressions, such
as the Bring radical (see Ex.~\ref{ex:Bring}). Second, our theoretical framework shows
that the general set of algebraic functions precisely characterizes the set of functions
that can be stably computed online by a CRN.
Third, the quadratization and lazy-negative optimization algorithms presented in this paper allow us to generate more concise CRNs.
On the example given in section 3.4 of~\cite{BEHL09al} 
for the quadratic expression $$y = \frac{b - \sqrt{b^2-4\diff{\cdot}a\diff{\cdot}c}}{2\diff{\cdot}a}$$ used to find the root of a
polynomial of second order, our compiler generates a CRN of $7$ species (including the $3$ inputs) and $11$
reactions, while their CRN following the syntax of the expression uses $10$ species and $14$ reactions.
Moreover, our generated CRN of smaller size includes dual-rail variables to give correct answers for the negative values of $y$.

\diff{Our theory may also present some
similarities with the field of linear filters. Filters are usually used to process signals that
oscillate around a rest value (usually $0$) and are usually studied through their
response to sinusoidal inputs. There are however two crucial differences between our theory
and linear filters. First, due to the presence of bi-molecular reactions, our systems are generally not linear. Second,
the main focus of linear filter theory is the transient behaviour, while we are more
interested by the equilibrium states of the system, and by the view of the equilibrium space
as a function from the inputs to the output.}

\section{Definitions and main theorem}

For this article, we denote single chemical species with lower case letters and set of
species with upper case letters, e.g.~$X = \{ x_1, x_2, \ldots \}$.  \diff{Moreover, we
will use $x$ for the input species, $y$ for the output and $z$ for the intermediate ones.}
By abuse of notation, we will use the same symbol for the variables of the ODEs, the
chemical species and their concentrations, the context being sufficient to remove any
ambiguity.

\subsection{Chemical reaction networks}

A chemical reaction with mass action law kinetics is composed of a multiset of reactants, a multiset of products and one rate constant.
Such a reaction can be written as follows:
\begin{equation}
	a+b \xrightarrow{k} 2\diff{\cdot}a
\end{equation}
where $k$ is the rate constant, and the multisets are represented by linear expressions in which the (stoichiometric) coefficients
give the multiplicity of the species in the multisets, here 2 for the product $a$, the coefficients equal to 1 being omitted.
In this example, the velocity of the reaction is the product
$k \diff{\cdot}a \diff{\cdot}b$, i.e.~the rate constant $k$ times the concentration of the reactants, $a$ and $b$.

In this paper, we consider \diff{elementary} CRNs with \diff{at most two reactant per reaction and} mass action law kinetics only. 
The Turing-completeness of this setting~\cite{FLBP17cmsb} shows that there is no loss of generality with that restriction.
It is \diff{also} well known that the other kinetics, such as Michaelis-Menten or Hill kinetics,
can be derived by quasi-steady state or quasi-equilibrium reductions of more complex CRNs
with mass action kinetics~\cite{Segel84book}.

\diff{However here, we study the case where some species, called \textit{environment input} species, are driven by the
  environment. Their concentrations thus do not depend on the CRN under study.
  It is worth noting that this notion of \textit{open CRN} with species concentrations imposed by the
  environment is of course instrumental in biological systems and their models,
  and has its own dedicated option in the Systems Biology Markup Language (SBML): \textit{boundaryCondition = true}~\cite{sbml20msb}.}

\begin{definition}
   The differential semantics of an \diff{open} CRN with a distinguished set of
   \diff{environment} input
	species $X$,
   is \diff{defined by} the usual ODEs for the output and intermediate species
   \diff{(respectively $f_y$ and $f_z$) and
	the time-function of the environment input $f_X(t)$.}
	\diff{\sout{
		$s \notin S^p$, and null differential functions for the \diff{\sout{pinned}driven} species}}
	\sout{$\forall s \in S^p,\quad \frac{d s}{dt} = 0.$}
   \diff{
   \begin{equation}\begin{aligned}
      X(t) &= f_X(t) \\
      \frac{dy}{dt} &= f_y(X,y,Z) \\
      \frac{dZ}{dt} &= f_z(X,y,Z)
   \end{aligned}\end{equation}
}
\end{definition}

\diff{\sout{This driving process may}This choice of modelling where the inputs are
pulled out of the influence of the CRN may have several explanations. It may} be due to a scale separation between the different concentrations (one
of the species is so abundant that the CRN essentially do not modify its
concentration), to a scale separation of volume (e.g. a compartment within a cell and a freely diffusive small molecule)
or to an active mechanism ensuring perfect adaptation (e.g. the
input is produced and consumed by some external reactions faster than the CRN itself, thus locking its concentration).

\subsection{Stabilization}

We are interested in the case where one particular species of the CRN, called the output,
\diff{\sout{is such that, whatever moves the inputs may do, once the
inputs are fixed, the concentration of the output species stabilizes on the result of some function
of the fixed inputs. Furthermore, we want this value to be robust to small
perturbations of both the auxiliary variables and the output. Of course, if the
inputs are modified, the output has to be modified.}  has a unique stable state for
every possible value of the inputs. This means that if the inputs are constant functions,
the output will \diff{always} converge to this stable state, hence ensuring robustness to
small perturbations of both the auxiliary and output variables.} The output thus encodes a
particular kind of robust computation of a function which we shall call
\textit{stabilization}.  \diff{Of course, if the inputs are not constant functions, it is
impossible for the output to follow exactly a moving target as it needs some time to
perform the online computation. We will nonetheless provide in section~\ref{sec:error} some guarantee
about the error of the output provided some regularity condition of the inputs: namely
that they are $k$-Lipschitz. We can then ensure that under some conditions and for closed
domain of $\D$, we have $|y(t) - f(X(t))| \leq p$ where the precision $p$ depends on $f$,
$k$ and a kinetic turnover parameter $\alpha$.}

\begin{definition}\label{def:stabilizes}
   We say that an \diff{open} CRN over a set of $m+1+n$ species $\{X,y,Z\}$ with 
\diff{\sout{pinned}environment} inputs $X$ of cardinality $m$ and distinguished outputs $y$, stabilizes the
function $f : I \to \R_+$, with $I \subset \R_+^m$, over the domain $\D \subset \R_+^{m+1+n}$
if:
\begin{enumerate}
  \item $\forall X^0 \in I$ the restriction of the domain $\D$ to the slice
$X=X^0$ is of plain dimension $n+1$, and
	\item $\forall (X^0,y^0,Z^0) \in \D$ the Polynomial Initial Value Problem (PIVP) given by the differential semantic with
		\diff{\sout{pinned}} constant input species \diff{$\forall t, X(t) = X^0$} and the
		initial conditions $\sout{X^0},y^0,Z^0$ is such that:
      $\lim_{t\rightarrow \infty} y(t) = f(X\diff{^0}).$
\end{enumerate}

This definition is extended to functions of $\R^m$ in $\R$ by dual-rail
encoding~\cite{ET89book, FLBP17cmsb}: for a CRN over the species $\{X^+, X^-, y^+, y^-, Z\}$ we ask that
$\lim_{t\rightarrow \infty} (y^+ - y^-)(t) = f(X^+ - X^-),$ for all initial
	conditions \diff{and constant inputs} in the validity domain $\D$.

Let $\F_S$ be the set of functions stabilized by a CRN.
\end{definition}

Several remarks are in order.
A first remark concerns the fact that we ask for a domain $\D$ of plain
dimension $n+1$, i.e.~non-null measure in $\mathbb{R}^{n+1}$. That constraint
is imposed in order to benefit from a strong form of robustness: there exists an
open volume containing the desired \diff{equilibrium} point such that it is the unique
attractor in this space. Hence in this setting, minor perturbations \diff{of $y$ and $Z$} are always
corrected. This requirement of an isolated \diff{equilibrium} point also impedes 
from hiding information in the initial conditions. The following example
illustrates the crucial importance of that condition on the dimension of the
domain $\D$

\begin{example}\label{ex:cosine}
The following PODE is constructed in such a way that $z_2$ goes exponentially to $x$ while $y$ and $z_1$ remain
equal to $\cos(z_2)$ and $\sin(z_2)$ respectively.
\begin{gather}
\begin{aligned} 
  \frac{dx}{dt} &= 0, \quad &x(t=0) = \text{input} \\
  \frac{dy}{dt} &= -z_1 \diff{\cdot}(x-z_2) \quad &y(t=0) = 1 \\
  \frac{dz_1}{dt} &= y \diff{\cdot}(x-z_2) \quad &z_1(t=0) = 0 \\
  \frac{dz_2}{dt} &= (x-z_2) \quad &z_2(t=0) = 0\\
\end{aligned}
\end{gather}
One might think that this
PODE stabilizes the cosine as we have $\lim y(t) = cos(x)$ for any value of $x$. 
But cosine is not an algebraic function, and indeed, the only requirement for
	this PODE to be at steady state is: $x = z_2$, meaning that there exist
	\diff{equilibrium}
points for any value of $z_1$ and $x$. So this PODE does not stabilize the cosine
function.
The reason is that the cosine computation is encoded in the initial state. It is only for the
domain where $y = \cos(z_2)$ and $z_1 = \sin(z_2)$ that the computation works, but this domain
is of null measure in $\R^3$ which breaks the first condition of Def.~\ref{def:stabilizes}.
\end{example}

A second remark is that \diff{\sout{
since the inputs are fixed in our semantics (they are by definition \diff{\sout{pinned}driven}
species),
the target of the output species which is the result $f(X)$ of some function $f$
is not a fluctuating goal: it is fixed by the initial conditions.} as the inputs may
vary with time, so can the desired output target: $f(X(t))$.}
In practice, what we ask is that the dynamics of the ODE for the slice of the
domain $\D$ defined by \diff{\sout{imposing}fixing} the inputs, have a unique attractor satisfying $y = f(X^0)$.
But as we do not impose any constraint on the other variables ($Z$),
we cannot speak of an \diff{equilibrium} point, since the dynamics on the other variables may
not stabilize (e.g.~oscillations, divergence, etc.).
\diff{We can nevertheless define a notion of partial equilibrium:

\begin{definition}
	$(X^\star,y^\star,Z^\star)$ is a partial equilibrium point of a system with
    distinguished species $y$ if, fixing the environment
	inputs at $\forall t, X(t) = X^\star$ the trajectory initialized at $y = y^\star,
	Z = Z^\star$ is such that $\forall t\in \R_+, y(t) = y^\star.$
\end{definition}}

\diff{\sout{We will nevertheless speak
of these object as partial stable point. If we start from a point
on this partial stable point we will have: $\forall t, y(t) = f(X^0)$.}}

A third remark is that our definition implies that \diff{\sout{apart from a transient
behaviour of characteristic time $\tau$, the whole system is constrained to live
in, or nearby, the subspace defined by $y = f(X)$}all the partial stable states of the
domain $\D$ verified the relation $y = f(X)$. This property is strongly reminiscent of the
notion of ACR, but instead of a constant value, it is a computational relation that
is verified for all steady states. }.
\diff{
  This leads us to propose an alternative definition based on a notion of Absolute Functional Robustness
  that we prove equivalent.
\begin{definition}
A CRN with distinguished inputs $X$, output $y$ and intermediate species $Z$ displays
Absolute Functional Robustness (AFR) of the function $f$ on the domain $\D$ if for all
choices of $X$ in \D, there exist a partial stable equilibrium point $(X,y)$ included in
an open subset of $\D$ and verifying: $y= f(X)$.
\end{definition}

\begin{proposition}
A CRN displays AFR of $f$ on $\D$ if and only if it stabilizes $f$ on $\D$.
\end{proposition}
\begin{proof}
Suppose that a CRN stabilizes $f$ on $\D$, then for each possible value of $X$, there is a
unique partial equilibrium point by definition of stabilization. This point verifies
$y=f(X)$ and given the requirement on the rank of $\D$ we can find an open interval of
$\D$ containing this point. Therefore the CRN displays AFR.

Suppose now that a CRN displays AFR for function $f$ in domain $\D$.
Let us fix some value $v_X\in\D$ for the environment inputs $X$.
By definition of AFR, there exists a partial stable point $(X,y)$ and this point
is unique because of the functional relation $y=f(X)$. As this point is stable and unique
in the slice of the domain $\D$ where $v_X$ belongs, it is the unique attractor of all trajectories
of the model. Since this is true for all $X$, the CRN does stabilize
function $f$.
\end{proof}
}

\diff{In other words, as long as we stay in the domain $\D$ of the definition, the system is
constantly trying to impose a relation with some characteristic time of relaxation
$T_\alpha$}.
What is interesting is that \diff{\sout{if} as soon as} the inputs are
\diff{\sout{themselves}} varying with a characteristic time that is slower than
\diff{$T_\alpha$}, the output will follow those variations, hence preserving our desired
property up to an error coming from the delay.
\diff{An open CRN stabilizing $f$ is thus effectively performing a form of robust online
  computation of this function.\sout{as long as the system stays in the domain $\D$}}.

\diff{It is worth remarking also that in} a
synthetic biology perspective, it is in principle possible to use a time-rescaling to
modify the value of \diff{\sout{$\tau$}$T_\alpha$}.  While a small
\diff{\sout{$\tau$}$T_\alpha$} allows for a faster adaptation, this usually comes at the
expense of a greater energetic cost as the proteins turn-over tends to increase.

The \diff{precise} study
of the error committed by such a system \diff{\sout{when the inputs are functions of
time}} is the subject of section~\ref{sec:error}.

\subsection{Algebraic curves and algebraic functions}

In mathematics, an algebraic curve (or surface) is the zero set of a polynomial
in two (or more) variables. It is a usual convention in mathematics to
speak indifferently of the polynomial and the curve it defines, seen as the same object.
For instance, $x^2+y^2-1$ is seen as the unit circle.

Any polynomial $P$ can be expressed as a product of irreducible polynomials,
i.e.~polynomials that cannot be further factorized, up to a 
constant $k$:
$$\displaystyle{P = k. \prod_{i = 1 \dots n} P_i^{a_i}}.$$
The $P_i$'s are called the components of $P$, and $a_i$ the multiplicity of
$P_i$.
We say that $P$ is in reduced form if all the components have multiplicity one, $\forall i\ a_i=1$.
This is justified by one important result of
algebraic geometry: in an algebraic closed field, such as the complex numbers $\C$, the set of
points of an algebraic curve given with their multiplicity, suffices to define the polynomial in reduced form.
This makes algebraic geometry an elegant and powerful theory. 

\diff{\sout{In a non-algebraically closed field such as $\R$, a polynomial may have no real
root. This difficulty is however irrelevant to us in this paper since we start
with an algebraic real function, thus assuming the existence of real roots. For
the purpose of this article, this fundamental result provides a canonical
correspondence between an algebraic real function and its polynomial of minimal
degree, i.e.~a polynomial in reduced form, up to a multiplicative factor.}}

\diff{In our computational setting, we need however a less ambiguous relation and the picture is slightly more
complicated. To a given algebraic function $f$, we can associate a canonical polynomial
up to a multiplicative factor $\alpha P_f(X,y)$ provided that we impose $P_f$ to be of
minimal degree (that is of multiplicity $1$).
\diff{In the other direction}, a given polynomial $P_f(X,y)$, even if its multiplicity is $1$
everywhere may have several
solutions for a given $X$ (think of the unit circle), hence possibly defining several
functions. To select a particular one and remove any ambiguity, it is however sufficient
to provide one point $(X, y)$ of any of this function.
Then, we can use the implicit function theorem to recover the whole function until we
reach a branch point, that is a point at which the number of solutions changes.
Our relation is thus far from being $1$ to $1$. To a given function there are an
infinity of choices for $\alpha$, and there are an infinity of points along a curve that
gives the same function.
Nevertheless, once given a polynomial and one solution point away from a branching point, any ambiguity disappears and
one can associate a unique  function:
}

\begin{definition}
A function $f:I \subset \R^m \to \R$ is algebraic if there exists
	a polynomial \diff{\sout{$P$} $P_f$} of $m+1$ variables such that:
\begin{equation}
\forall X \in I, P_f(X, f(X)) = 0,
\end{equation}
	\diff{and $\forall X \in I$, $(X,f(X))$ is a point of multiplicity 1.}

\end{definition}

Let us denote by $\F_A$ the set of real algebraic functions.
We shall prove the following central theorem:
\begin{theorem}\label{thm:main}
 The set of functions stabilized by a CRN with mass action law kinetics
 is the set of algebraic real functions:
 $\F_S = \F_A.$
\end{theorem}

\diff{\sout{One technical difficulty comes from the fact that it is not immediate to determine the
function $f$ from the polynomial. Indeed for a given polynomial fixing the value of the
inputs results in one, several or no possible value for the output. Hence, a given
polynomial actually defines several functions on the domain of its input.
This is for instance the case of the unit circle curve defined by $x^2+y^2-1$.
If we see it as an equation to solve upon $y$, it admits
two solutions when $x \in ]-1,1[$, exactly one for $x=-1$ or $x=1$,
and no solution for other values of $x$. Hence, that curve defines two continuous
functions $y(x)$, each of them with support $]-1,1[$.

To overcome that difficulty, let us call branch point (or branch curve), the set of points
where the number of real roots of an algebraic function changes (${-1,1}$ in the previous
example). Now for a polynomial $P(X,y)$ and a given root $X,y$ that is not a branch point,
the implicit function theorem ensures the existence and uniqueness of the implicit function
up to the next branch point/curve.}}

\begin{example}\label{ex:unitcircle}
 The branch points of the unit circle polynomial $x^2+y^2-1$ are ${(-1,0),(1,0)}$.
 If we provide an additional point on the curve, e.g.~$(0,1)$,
 one can define the function that contains it and that goes from one branch point to the
other one, here:
\begin{align*}
  ]-1,1[ &\rightarrow \R \\
  f: x &\mapsto \sqrt{x^2-1}
\end{align*}
Fig.~\ref{fig:circle} \diff{detailed} in a latter section \diff{\sout{illustrates} shows} the flow diagram used in this
	example by \diff{the CRN synthesized by our} compiler to stabilize that function.

Similarly in the case of the sphere defined by the polynomial $x_1^2+x_2^2+y^2-1$,
the branch curve is the whole unit circle contained in the plane $y=0$. And
giving the point $0,0,-1$ is enough to define the whole surface corresponding to
the down part of the sphere inside the branch-curve circle.
\end{example}

\section{Proof}

\begin{lemma}\label{lmm:<-}
$\F_A \subset \F_S.$
\end{lemma}

\begin{proof}
Suppose that $f:I \to \R$ is an algebraic real function and let $\alpha P_f$ denote the canonical
   polynomial \diff{(that is the polynomial of lower
   degree)along with its multiplicative constant $\alpha \in \R^+$.}
   We have $\forall X \in I, \diff{\alpha} P_f(X,f(X)) = 0$. Let us choose a vector
$X^\star$ in the domain of $f$.

	Then, \diff{when the inputs are $\forall t, X(t) = X^\star$} the PODE
	\diff{
	\begin{equation}\label{eq:stab}
		\frac{dy}{dt} = \pm \alpha P_f(X(t),y),
	\end{equation}
	}
	is such that \diff{\sout{$y=f(X)$} $y = f(X^\star)$} is an \diff{equilibrium} point. By choosing the sign such that,
	locally \diff{$\pm P_f(X^\star,y)$} is negative above \diff{$y = f(X^\star)$} and positive below, we ensure
that this point is stable.

	The fact that the polynomial has to change the sign across the \diff{equilibrium} point is
	due to the fact that \diff{\sout{we choose the polynomial of minimal degree} the
	canonical polynomial is of minimal degree}, hence it
has to be in reduced form and the multiplicity of every branch of the curve is
one: the sign cannot be the same on both sides of the curve.

	It is worth remarking that any ODE system made of elementary mathematical
	functions can be transformed in a polynomial ODE system~\cite{HFS21cmsb}, hence one can wonder why we restrict here
	to polynomial expressions. This comes from the condition that asks that the
	domain $\D$ be of plain dimension in Def.~\ref{def:stabilizes}.
    The polynomialization of an ODE system may indeed introduce constraints
	between the initial concentrations which is precisely what is forbidden by the
	requirement upon $\D$.

	\diff{
\sout{
Now, let us note $Y^+ = \inf(Y\ |\ P(X,Y)=0, Y>f(X))$ and $Y^- = \sup(Y\ |\ P(X,Y)=0, Y<f(X))$,
with $\pm \infty$ values if the set is empty.
We know that for all $y$
in $]Y^-, Y^+[$, the only attractor is $f(X)$ and as a polynomial can only have
a finite number of zeros, those sets are non empty.
}

	Now, let us note $y^\text{sup} = \min(y\ |\ P(X^\star,y)=0, y>f(X^\star))$ and
	$y^\text{inf} = \max(y\ |\ P(X^\star,y)=0, y<f(X^\star))$, with $\pm \infty$ values if
	the sets are empty. We know that $]y^\text{inf}, y^\text{sup}[$ is not empty because a
	polynomial can only have a finite number of zeros. Moreover, for all $X^\star$, the only
	attractor in the interval verifies $y = f(X^\star)$. This concludes the proof if we
    simply want to construct a PODE.
	}

    \diff{If we want to construct an elementary CRN, we have to tackle two difficulties, similarly to~\cite{FLBP17cmsb}.
    First the variables may become negative while concentration cannot. And second, we
    have to quadratize the PODE in order to obtain elementary reactions.}
For all variables that are not bound to be positive, the dual-rail encoding
consists in splitting the variable into two positive variables corresponding to
the positive and negative parts $x = x^+ - x^-$. Then, all positive monomials
	can be dispatched to the \diff{derivative of the} positive part and all negative ones to the
	\diff{derivative of the} negative
	part (but with a positive sign). \diff{\sout{with the addition of a}We finally add the}
   catalytic degradation \diff{$x^+ + x^- \xrightarrow{\text{fast}} \emptyset$}, with a
rate: \diff{$\text{fast}$} elaborated to ensure that the product $x^+ x^-$ is bounded by a fixed
constant, this imply that one of the two variables ($x^+$ or $x^-$) stays
close to zero as described in~\cite{FLBP17cmsb} \diff{(Definition 8 and Theorem 3.)}.
\diff{\sout{This construction do}Crucially, this construction does} not modify
	the behavior of the variable $x$, ensuring that all stable \diff{\sout{fixed}} points of the
        original system are preserved.
        \sout{It is worth noting that that dual-rail encoding
may be necessary even for positive functions when the auxiliary variables can take
negative values.}

\end{proof}

\begin{lemma}\label{lmm:->}
$\F_S \subset \F_A.$
\end{lemma}

\begin{proof}
	Let us suppose that $f$ is a function stabilized by a mass action CRN \diff{interpreted
	through the differential semantics}.
	The idea is to use the characterization of functions that are \diff{\sout{projectively
	polynomial.} solution of a
   Polynomial Initial Value Problem given in~\cite{CPSW05ejde} (called projectively
   polynomial in the vocabulary of the paper).
	Indeed, Theorem~3 in that paper states that if a function $u$ is the solution of a PIVP, then there
	exist $n \in \mathbb{N}$ and a polynomial $Q$ of $n+1$ variables such that:
	$Q(u, u', \ldots , u^{(n)}) = 0$. The proof rely on the successive elimination of the
	other variables of the PODE until the equations are only about $u$ and its derivatives. If
	$n$ variables are thus eliminated, $Q$ will have (at most) $n+1$ variables.}

	\diff{\sout{By using the higher-order derivatives of
the stabilized variable, it is shown in~\cite{CPSW05ejde} that one can eliminate all the auxiliary
	variables and}By invoking this theorem on the output variable, to eliminate all the
	intermediate variables, we} obtain a single equation of the form:
	\begin{equation} \label{eq:y_charac}
		P(X, y, y^{(1)}, \ldots, y^{(n)}) = 0.
	\end{equation}
	\diff{Here $P$ is the characterization of $y$ extracted from the system but not yet the
	polynomial defining the algebraic curve of $f$.}

	\diff{\sout{Using the fact that for all $X$, $y=f(X)$ is a partial stable point by definition, if we use it
	as initial condition we immediately get:} Then, the idea is to use the fact that for
	all $X \in I$, $f(X)$ is a partial equilibrium point of the system, that is:}
\begin{align*}
  X &= X, \\
  y &= f(X), \\
	y^{(k)} &= 0 \quad \forall k \in [1,n].
\end{align*}
	\diff{(Recall that we cannot use the term of equilibrium point as we do not know the
	behaviour of the intermediate species.)}

	\diff{\sout{Injecting this in the characterization of the function $y$} Injecting this
	definition of a partial equilibrium point in the characterization of $f$
	(equation~(\ref{eq:y_charac}))}, we obtain:
\begin{equation}
  \forall X, P^\star(X, f(X)) = 0,
\end{equation}

There are now two cases. Either $P^\star$ is not trivial and effectively
	defines the surface of \diff{equilibrium} points: this gives a polynomial for $f$, hence $f$
is algebraic. Either $P^\star$ is the uniformly null polynomial. But in this
	case, every points in the $X,y$ plane may be a \diff{stable} point and the domain $\D$
of the definition of stabilization is reduced to a single point, yet we asked it
to be of non-null measure. Therefore, $P^\star$ is not trivial and $f$ is algebraic.
\end{proof}

\section{Compilation pipeline for generating stabilizing CRNs} \label{seq:compilation}

The proof of lemma~\ref{lmm:<-} is constructive and provides a method 
to transform any algebraic function defined by a polynomial and one point,
in an abstract CRN that stabilizes it.
This is implemented with a
command

\verb$stabilize_expression(Expression, Output, Point)$\\
with three arguments:
\begin{description}
\item[{\tt Expression}:] For a more user friendly interface, we accept in input more
	general algebraic
expressions than polynomials; the non polynomial parts are detected
and transformed by introducing new variable/species to compute their values;
\item[{\tt Output}:] a name of the Output species different from the input;
\item[{\tt Point}:] a point on the algebraic curve that is used to
	determine the branch of the curve to stabilize \diff{\sout{if several exist} by
		choosing the sign in front of $\frac{dy}{dt}$}.
\end{description}

Similarly to our previous pipeline for compiling any elementary function in an
abstract CRN that computes it~\cite{HFS21cmsb,HFS20cmsb,FLBP17cmsb}, 
our compilation pipeline for generating stabilizing CRNs follows the same sequence of symbolic transformations: 
\begin{enumerate}
\item polynomialization
\item stabilization
\item quadratization
\item dual-rail encoding
\item CRN generation
\end{enumerate}
yet with some important differences.

\subsection{Polynomialization}

This optional step has been added just to obtain a more
user friendly interface, since polynomials may sometimes be cumbersome to manipulate.
The first argument thus admits algebraic expressions instead of
being limited to polynomials.

The rewriting simply consists in detecting all the non-polynomial terms of the form
$\sqrt[a]{b}$ or $\frac{a}{b}$ in the initial
expression and replace them by new variables, hence obtaining a polynomial.

Then to compute the variables that just have been introduced, we perform
the following basic operations on each of them to recover polynomiality:
\begin{align*}
n = \sqrt[a]{b} &\mapsto n^a-b \\
n = \frac{a}{b} &\mapsto n\diff{\cdot}b-a
\end{align*}
and recursively call \verb|stabilize_expression| on these new expressions
with the introduced variable (here $n$) as desired output.

\subsection{Stabilization}
To select the branch of the curve to stabilize,
it is sufficient to choose the sign in front of the polynomial in equation~(\ref{eq:stab}).
such that at the designated point, the second derivative
of $y$ is positive. For this, we use a formal derivation to compute the sign of
the polynomial, and \diff{reverse the sign} if necessary.

\subsection{Quadratization}

The quadratization of the PODE is an optional transformation which aims at generating
elementary reactions, i.e.~reactions having at most two reactants each,
that are fully decomposed and more amenable to concrete implementations with real biochemical molecular species.
It is worth noting that the quadratization problem to solve
here is a bit different from the one of our original pipeline studied
in~\cite{HFS20cmsb} since we want to preserve a different property of the PODE.
It is necessary here that the introduced variables stabilize on their target value
independently of their initial concentrations.
While it was possible in our previous framework to initialize the different
species with a precise value given by a polynomial of the input, this is no more the case here as
the domain $\D$ has to be of plain dimension. 

The variables introduced by quadratization correspond to monomials of order higher than $2$ that can
thus be separated as the product of two variables corresponding to monomials of
lower order: $A$ and $B$. Those variables can be either present in the original
polynomial or introduced variables.
The following set of reactions:
$$A+B \rightarrow A+B+M$$
$$M \rightarrow \emptyset,$$
gives the associated ODE:
\begin{equation}\frac{dM}{dt} = A\diff{\cdot}B-M, \label{eq:monomial_computation}\end{equation}
for which the only stable point satisfies: $M = A\diff{\cdot}B$.

Furthermore as before, we are interested in computing a quadratic PODE of minimal dimension.
In~\cite{HFS20cmsb}, we gave an algorithm in which the introduced variables were always equal to the monomial they compute,
whereas in our online stabilization setting, this is true only when $t \rightarrow \infty$.
For instance, if we replaced $A\diff{\cdot}B$ by $M$ in
equation~(\ref{eq:monomial_computation}), the system would no longer adapt to changes of the input.
To circumvent this difficulty, it is 
possible to modify the PIVP and use it as input of our previous algorithm to
take this constraint and still obtain the minimal set of variables.
In our previous computation
setting, the derivatives of the different variables where simply the
derivatives of the associated monomials computed in the flow generated by the
initial ODE.
In Alg.~\ref{algo:hack}, we construct \diff{an auxiliary ODE} containing twice as many
variables, the derivatives of which being built to ensure
that the solution is correct. The idea is that the actual variables are of the
form $Mb$ and the $Mb^2$ variables exist only to construct the solution. To compute
quadratic monomials with a
$b^2$ term present in the derivatives of the $Mb$ variables (the
``true'' variables),
one can either add two $Mb$ variables to the solution set or add
a single $Mb^2$ variable. As can be seen on the lines 5 and 9 of Alg.~\ref{algo:hack}, $Mb^2$ variables
require that the corresponding $Mb$ is in the solution set.

\begin{algorithm}
\caption{Quadratization algorithm for a PODE stabilizing a function.
The $minimal\_quadratic\_set(PODE, y)$ returns the
minimal set of variables containing $y$ sufficient to express all its derivatives 
in quadratic form~\cite{HFS20cmsb}.\label{algo:hack}}
\textbf{Input}: A PODE of the form $\frac{dx_i}{dt} = P_i(X)$, with $i \in [1,n]$ to compute $x_n$\\
\textbf{Output}: A set $S$ of monomials to quadratize the input.
\begin{algorithmic}[5]
\State $ODE_\text{aux} \gets \emptyset$
\State find an unused variable name: $b$
\ForAll{$i \in [1,n]$}
  \State add $\frac{dx_ib}{dt} = P_i(X) \diff{\cdot} b^2$ to $ODE_\text{aux}$
  \State add $\frac{dx_ib^2}{dt} = x_ib$ to $ODE_\text{aux}$
\EndFor
\State $AllMonomials \gets$ the set of monomials that are less or equal to a
monomial present in one of the $P_i$ and not in $X$.
\ForAll{$M \in AllMonomials$}
  \State add $\frac{dMb}{dt} = Mb^2$ to $ODE_\text{aux}$
  \State add $\frac{dMb^2}{dt} = Mb$ to $ODE_\text{aux}$
\EndFor
\State $S_\text{aux} \gets minimal\_quadratic\_set(ODE_\text{aux}, x_n b)$
\State $S \gets \emptyset$
\ForAll{$Mb \in S_\text{aux}$}
  \State add $M$ to $S$
\EndFor
\State \textbf{return} $S$
\end{algorithmic}
\end{algorithm}

This variant of the quadratization problem studied in~\cite{HFS20cmsb} has the same
theoretical complexity, as shown by the following proposition:

\begin{theorem}
The quadratization problem of a PODE for stabilizing a function and minimizing the number of variables is NP-hard.
\end{theorem}
\begin{proof}
The proof proceeds by reduction of the vertex covering of a graph as in~\cite{HFS20cmsb},
	theorem~2.
Let us consider the graph $G=(V,E)$ with vertex set $v_i, i \in [1,n]$ and edge set $E \in V\times V$.
And let us study the quadratization of the PODE with input variables $V \cup
\{a\}$ and output variable $y$ such that the $y$ computes
$\sum_{v_i v_j \in E} v_i \diff{\cdot}v_j \diff{\cdot}a$.
	The (only) derivative is:
	\begin{equation}\label{eq:proof_NP}
\frac{dy}{dt} = \sum_{v_i v_j \in E} v_i \diff{\cdot}v_j \diff{\cdot}a - y.
\end{equation}
	\diff{\sout{An optimal quadratization contains variable corresponding either to $v_i a$ or
$v_i v_j$ indicating that an optimal covering of the graph $G$ contains either
the node $v_i$ either indifferently $v_i$ or $v_j$. Hence en optimal
	quadratization gives us an optimal covering which concludes the proof.}
	The idea is then to derive an optimal covering of $G$ from a quadratization of
	(\ref{eq:proof_NP}) minimizing the number of species.
	That quadratization is composed of monomials of the form $v_i a$ or $v_i v_j$.
	Let us start with an empty set $S$, for each monomial of the form $v_i a$ we add the
	nodes $v_i$ to the set while for each monomial of the form $v_i v_j$ we add either
	$v_i$ or $v_j$ to $S$. Then $S$ defines an optimal covering of $G$ as by construction
	of (\ref{eq:proof_NP}), each edge has one of its nodes in $S$ and $S$ is minimal
	otherwise there would exist a better quadratization. This proves that
	there exist a reduction from vertex covering to the quadratization and thus that
	finding an optimal monomial quadratization is NP-hard.}
\end{proof}

\diff{It is worth remarking that the proof no longer works if the introduction of variables corresponding to
polynomials is allowed. Indeed the system (\ref{eq:proof_NP}) can be
quadratized with a single new variable: $\sum_{v_i v_j \in E} v_i \diff{\cdot}v_j$. The
complexity of quadratization using variables for polynomials is still an open problem.}

\diff{The encoding of this problem in MAXSAT given in~\cite{HFS20cmsb} and the solution preserving heuristics described in~\cite{HFS21cmsb}
still here} with the slight modification
mentioned above concerning the introduction of new variables. 
Alg.~\ref{algo:hack}, when invoked with an optimal search for
$minimal\_quadratic\_set$, is nevertheless not guaranteed to generate optimal solutions,
because of the \diff{\sout{pseudo}auxiliary} variables noted $Mb^2$.
Despite those theoretical limitations, the CRNs generated by Alg.~\ref{algo:hack} are particularly concise,
as shown in the example section below and already mentioned above for the compilation of algebraic expressions compared to~\cite{BEHL09al},

\subsection{Lazy dual-rail encoding}
\diff{Our CRN encodes variables using the concentrations of chemical species that are, by
nature, positive quantities.}
As in our original compilation pipeline~\cite{FLBP17cmsb},
it is \diff{thus} necessary to modify our PODE in order to impose that no variable may
become negative. This is possible through a lazy version of dual-rail encoding.
First by detecting the variable that are or may become negative
and then by splitting them between a positive and negative part, thus
implementing a dual-rail encoding of the variable: $y = y^+ - y^-$.
Positive terms of the
original derivative are associated to the derivative of $y^+$ and negative terms
to the one of $y^-$ and a fast mutual degradation term is finally associated to both
derivative in order to impose that one of them stays close to zero~\cite{FLBP17cmsb}.

\subsection{CRN generation}
The same back-end compiler as in our original pipeline
is used,
i.e.~introducing one reaction for each monomial.
It is worth remarking that this may have for effect to aggregate in one reaction
several occurrences of a same monomial in the ODE system~\cite{FGS15tcs}.

\section{Examples}

\begin{example}\label{ex:circle}
	As a first example, we can study the unit circle: $x^2+y^2-1$. \diff{The stabilizing
	CRN for the upper part of the circle can be obtain in Biocham by invoking the following command:}
   \begin{center}
\begin{verbatim}stabilize_expression(x^2+y^2-1, y, [x=0, y=1]).
\end{verbatim}
   \end{center}

	Our pipeline gives us \diff{\sout{for the upper part of the circle,}} the following CRN.
\begin{gather}
\begin{aligned}
  \emptyset &\rightarrow y^+ &\quad
2 \diff{\cdot}y^- &\rightarrow 3 \diff{\cdot}y^- \\
  2 \diff{\cdot}x &\rightarrow y^- + 2 \diff{\cdot}x &\quad
2 \diff{\cdot}y^+ &\rightarrow y^- + 2 \diff{\cdot}y^+ \\
y^+ + y^- &\xrightarrow{\text{fast}} \emptyset
\end{aligned}
\end{gather}
the flow of the PODE associated to this model can be seen in figure~\ref{fig:circle}\textbf{A} and the
steady state is depicted in figure~\ref{fig:circle}\textbf{B} as a function of
$x$ in the positive quadrant.

\end{example}

\begin{figure}
\centering
\textbf{A.}\includegraphics[width=0.4\textwidth]{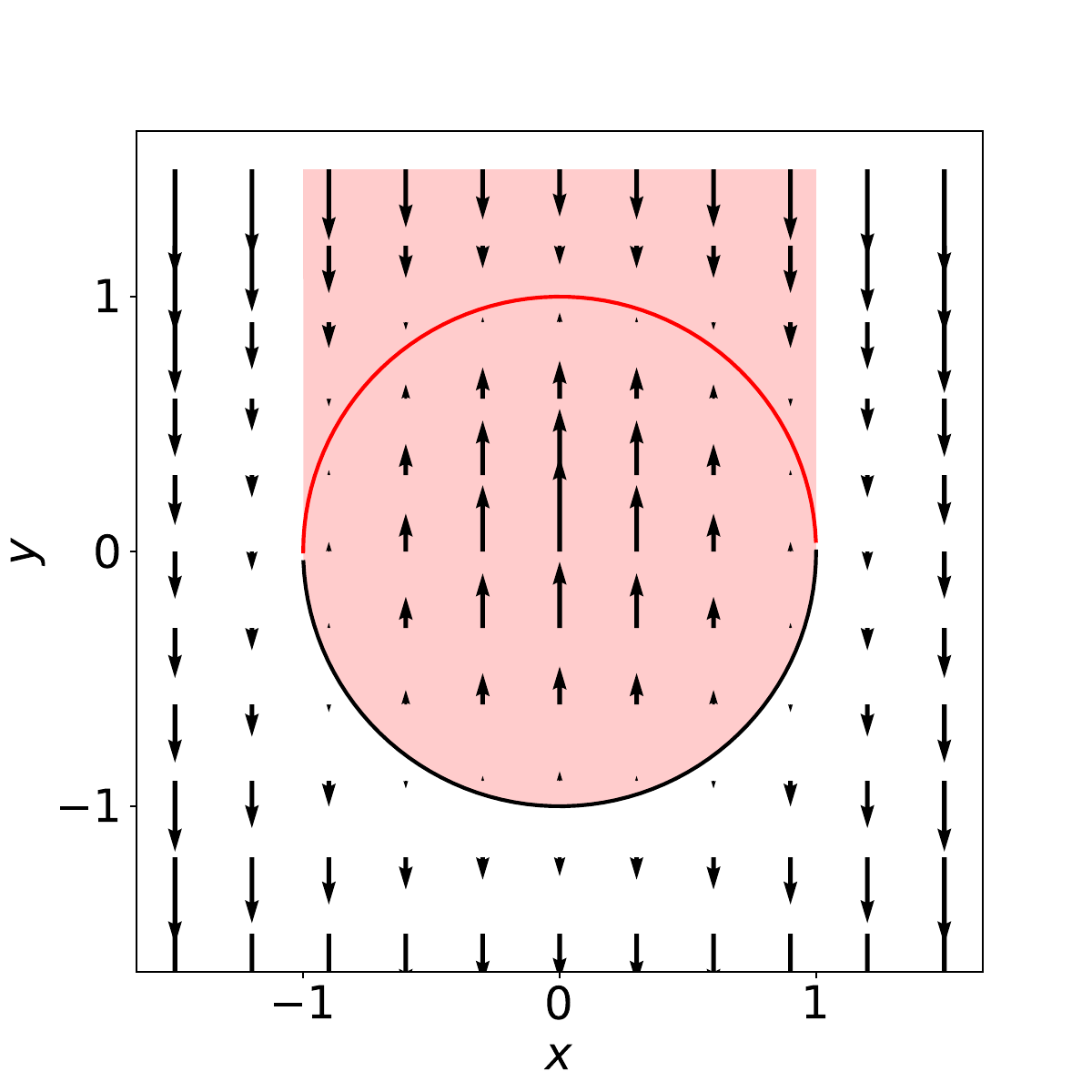}
\textbf{B.}\includegraphics[width=0.5\textwidth]{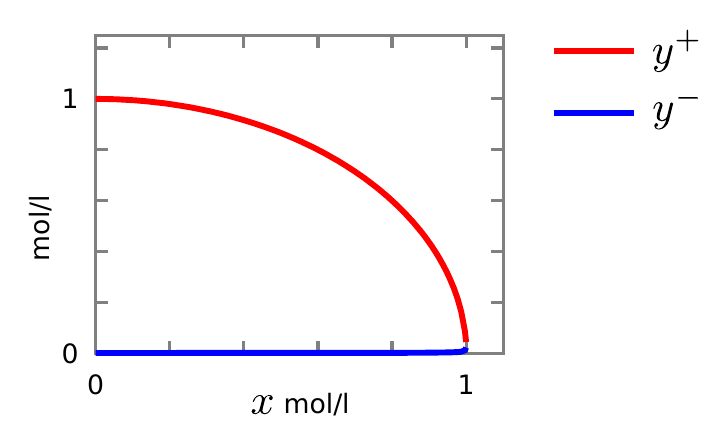} \quad
\caption{
\label{fig:circle}
\textbf{A.} Flow diagram in the $x,y$ plane before the dual rail encoding for
the stabilization of the unit circle. The arrows indicate the direction and strength
of the flow. The upper part of the curve (in red) indicates the stable branch
of the system and we colored in light red the domain $\D$ in which the system
will reach the desired steady state. Outside of $\D$, the system is driven to the
divergent state $\lim y \rightarrow -\infty$.
\textbf{B.} Dose response diagram of the generated CRN where the input
concentration ($x$) is gradually increased from $0$ to $1$ while recording the steady
state value of the output species $y^+, y^-$.
}
\end{figure}

\begin{example}\label{ex:serpentine}
Even rather involved algebraic curves need surprisingly few species and
	reactions \diff{to be stabilized}.
This is the case of the serpentine curve, or anguinea,
defined by the polynomial
$(y-2) \left((x-10)^2+1\right) = 4 (x-10)$ for which we choose the point $x=10, y=2$
to enforce stability.
	The compilation process takes less than 100ms on a typical laptop\footnote{An Ubuntu 20.04, with
	an Intel Core i6, $2.4$GHz x $4$ cores and $15.5$GB of memory.}.
The generated CRN reproduces the anguinea curve on the $y$ variable, as shown in Fig.~~\ref{fig:anguinea}.
It is composed of the following $4$ species and
$12$ reactions:

\begin{gather}
\begin{aligned}
\label{model:anguinea}
  y_m+y_p &\xrightarrow{\text{fast}} \emptyset, &\quad
 2 \diff{\cdot}x &\rightarrow z+2 \diff{\cdot}x,\\
  z &\rightarrow \emptyset, &\quad
 \emptyset &\xrightarrow{162} y_p,\\
  y_m &\xrightarrow{101} \emptyset, &\quad
 x+y_p &\xrightarrow{20} x+2 \diff{\cdot}y_p,\\
  z+y_m &\rightarrow z, &\quad
 z &\xrightarrow{2} z+y_p,\\
  x &\xrightarrow{36} x+y_m, &\quad
 y_p &\xrightarrow{101} \emptyset,\\
  x+y_m &\xrightarrow{20} x+2 \diff{\cdot}y_m, &\quad
 z+y_p &\rightarrow z.
\end{aligned}
\end{gather}
\end{example}

\begin{figure}
\centering
\includegraphics[width=0.7\textwidth]{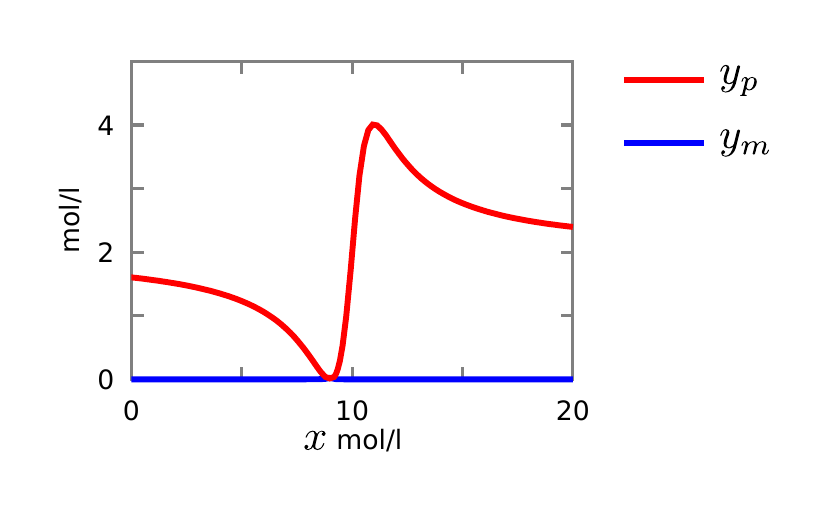}
\caption{\label{fig:anguinea}
	Dose-response diagram of the \diff{model (\ref{model:anguinea}), the} CRN generated by compilation of the serpentine algebraic curve,
 $(y-2) \left((x-10)^2+1\right) = 4 (x-10)$,
 with $x$ as input and $y$ as output.
}
\end{figure}

\begin{example}\label{ex:Bring}
In the field of real analysis, the Bring radical of a real number $x$ is defined
as the unique real root of the polynomial: $y^5+y+x$. The Bring radical is an
algebraic function of $x$ that cannot be expressed by any algebraic expression.

The stabilizing CRN generated by our compilation pipeline is composed of
$7$ species (${y_m, y_p, y2_m, y2_p, y3_m, y3_p, x}$),
\diff{where $y2$ (resp. $y3$) is
	the variable introduced to stabilized $y^2$ (resp. $y^3$) and
  the $p$ and $m$ subscripts represent the positive and negative parts of a signed variable,} and $20$ reactions presented in
  model~\ref{model:Bring}.
A dose-response diagram of that CRN
is shown in Fig.~\ref{fig:bring}.

\begin{gather}
\begin{aligned}
\label{model:Bring}
  y_m+y_p &\xrightarrow{\text{fast}} \emptyset, &\quad
  y2_m+y2_p &\xrightarrow{\text{fast}} \emptyset, \\
  y3_m+y3_p &\xrightarrow{\text{fast}} \emptyset, &\quad
  y_p &\xrightarrow{} \emptyset, \\
  y2_m+y3_p &\xrightarrow{} y2_m+y3_p+y_p, &\quad
  y2_p+y3_m &\xrightarrow{} y2_p+y3_m+y_p, \\
  x &\xrightarrow{} x+y_m, &\quad
  y_m &\xrightarrow{} \emptyset, \\
  y2_p+y3_p &\xrightarrow{} y2_p+y3_p+y_m, &\quad
  y2_m+y3_m &\xrightarrow{} y2_m+y3_m+y_m, \\
  2 \cdot y_p &\xrightarrow{} y2_p+2 \cdot y_p, &\quad
  2 \cdot y_m &\xrightarrow{} y2_p+2 \cdot y_m, \\
  y2_p &\xrightarrow{} \emptyset, &\quad
  y2_m &\xrightarrow{} \emptyset, \\
  y2_p+y_p &\xrightarrow{} y2_p+y3_p+y_p, &\quad
  y2_m+y_m &\xrightarrow{} y2_m+y3_p+y_m, \\
  y3_p &\xrightarrow{} \emptyset, &\quad
  y2_p+y_m &\xrightarrow{} y2_p+y3_m+y_m, \\
  y2_m+y_p &\xrightarrow{} y2_m+y3_m+y_p, &\quad
  y3_m &\xrightarrow{} \emptyset.
\end{aligned}
\end{gather}
\end{example}
   
\begin{figure}
\centering
\includegraphics[width=0.7\textwidth]{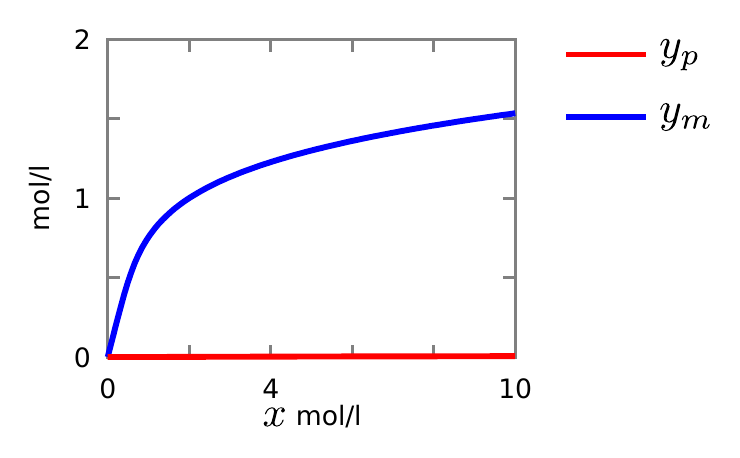}\\
\caption{
\label{fig:bring}
	\diff{Dose response of the model (\ref{model:Bring}).}
The bring radical is the real root of the polynomial equation $y^5+y+x=0$. As
this quantity is negative, the output is read on the negative part of the output
$y_m$. It is the
simplest equation for which there is no expression for $y$ as a function of $x$.
}
\end{figure}

\begin{example}\label{ex:hill5}
To generate the CRN that stabilize the Hill function of order 5, we can use the
expression $y-\frac{x^5}{1+x^5}$ along with the point $x=1, y=\frac{1}{2}$.
Our compilation pipeline generates the following model with $6$ species and $10$
	reactions:\\
	\begin{center}
\begin{tabular}{ll}
	$2 \diff{\cdot}x \rightarrow z_1+2 \diff{\cdot}x$, &
   $2 \diff{\cdot}z_1 \rightarrow z_2+2 \diff{\cdot}z_1$, \\
	$x+z_2 \rightarrow x+z_2+z_3$, &
	$z_4+z_3 \rightarrow z_3+y$, 
\end{tabular}\\
\begin{tabular}{lll}\label{model:hill5}
	$\emptyset \rightarrow z_4$, &
   $z_1 \rightarrow \emptyset$, &
	$z_2 \rightarrow \emptyset$, \\
	$z_3 \rightarrow \emptyset$, &
	$z_4 \rightarrow \emptyset$, &
	$y \rightarrow \emptyset$,
\end{tabular}
	\end{center}
all kinetics being mass action law with unit rate.
The $z's$ are auxiliary variables corresponding to the following expressions:
$$z_1 = x^2, \quad
z_2 = x^4, \quad
z_3 = x^5, \quad
z_4 = \frac{1}{1+x^5}.
$$

The production and degradation of $z_4$ may be surprising, but looking at all
the reactions implying both $z_4$ and $y$, we can see that their sum follow the
equation $\frac{d(z_4+y)}{dt} = 1-(z_4+y)$ hence ensuring that the sum of the
	two is \diff{\sout{fixed independently of} driven back to $1$, independently of} their initial concentrations.
It is worth remarking that another way of
reaching the same result would be to directly build-in conservation laws into
our CRN, hence using both steady state and invariant laws to define our steady
state which however would make us sensitive to the initial concentrations.
\end{example}

\section{Error control}\label{sec:error}

An important aspect of our framework that we still do not have talked about is the ability
to determine the error between the target function and the output of our synthesized CRN.
Upon this section, we will adopt the following point of view. \diff{On the one hand}, the system is
supposed to be noiseless, meaning that both the auxiliary and the output variables exactly
follow the PODE without any perturbations. \diff{On the other hand}, the inputs are supposed to
be driven by the environment in a way that we do not know in advance, thus fully
determining the target function $f(X(t))$. We would like to have an estimation of the
error between the output of the CRN and the target function\diff{:
\begin{equation}
   \epsilon(t) = y(t) - f(X(t)).
\end{equation}}

Of course, this is possible only \diff{\sout{with} by assuming} some regularity property of the input.
\diff{\sout{In particular, if we do not ask for a continuous input function it is impossible for an
output that is continuous by construction to follow a target that may jumps
discontinuously. This may gives arbitrarily large and irrelevant errors.}
For example, if we work with non-continuous input functions, just after a singularity of the
inputs, the error may be arbitrarily large. But this does not seems fair as our output, as
a solution of a PODE, will always be continuous. But even continuous functions are too
liberal, since one can construct function that goes arbitrarily fast from one value to
another, once again making our CRN unable to follow the pace.}

To avoid this, we
shall suppose that the inputs have a given maximal
rate of change, or to say it mathematically, that all the input functions are
$k$-Lipschitz for a given $k \in \R^+$:
\diff{$\forall t_1, t_2, |X(t_2) - X(t_1)| \leq k |t_2 - t_1|$}.
\diff{\sout{Hence we have} Lipschitz functions are differentiable almost everywhere, and
for each time $t$ where the inputs are differentiable, we have}:
   \sout{$\forall i, \left| \frac{dx_i}{dt}(t) \right| \leq k.$}
\begin{equation}
   \diff{\left| \frac{dX}{dt}(t) \right| \leq k.}
\end{equation}

\diff{While this may seem a stringent requirement, it is actually quite natural
to suppose that the concentration of chemical species
cannot increase or decrease arbitrarily fast.}

\subsection{General case}

For the sake of clarity, we shall derive the equations for a function of one variable
\diff{and shall suppose without loss of generality that the input is differentiable
everywhere.} 

\diff{To generalize to the case of functions of several variables, it is sufficient to
sum the derivatives of $f$ with respect to all variables, thus providing an over approximation
of the derivative with respect to any unitary vector.}

\diff{To handle cases where the inputs is not differentiable everywhere, the argument we
give here can be derived on every interval where the input is differentiable and the
whole solution can be recovered by continuity.}

\diff{\sout{We moreover will} We shall also} focus on the mathematical framework of PODE where the output may be
directly linked to the inputs and will study the case where we have to introduce variables
to quadratize the polynomial later.
Hence, given an input $x(t)$ and a CRN stabilizing the output $y$ on the function $f$. By
definition, we have $P_f$ the polynomial of the curve $f$. We thus have:
\begin{equation}\label{eq:stabilizing_ODE}
   \frac{dy}{dt} = \pm \alpha \diff{\cdot}P_f(\diff{x(t)},y),
\end{equation}
where $\alpha \in \R^+$, is a multiplicative constant that do not modify the stabilized
function but allows for a faster convergence by increasing the kinetic turnover of the
CRN. \diff{It is worth noting that due to the time varying inputs,
equation~\ref{eq:stabilizing_ODE} is technically a non-autonomous ODE, without modifying our reasoning.}

We define the error as $\epsilon(t) = y(t) - f(x(t))$. Note that we define a signed error
and allow $\epsilon$ to be negative. Differentiating with respect to time, we obtain:
\begin{equation} \label{eq:dt_epsilon}
   \frac{d \epsilon}{dt} = \pm \alpha \diff{\cdot}P_f\left(x\diff{(t)}, f(x)+\epsilon \right)
   - \frac{df}{dx} \diff{\cdot}\frac{dx}{dt}.
\end{equation}

The main unknown in this equation is $P_f$ but, let us remark that by its very
definition, \diff{\sout{this polynomial}$P_f(x, f(x)+\epsilon)$} is null when $\epsilon =
0$. We can thus introduce a new
polynomial that will be important for our error control:
\begin{equation}
   \pm P_f\left( x, f(x)+\epsilon \right) = - \epsilon \diff{\cdot}Q_f(x, \epsilon).
\end{equation}
Note that this is not an approximation but a definition of $Q_f$. A good property of
$Q_f$ is that it has to be positive in order to stabilize the desired branch, this explain
why the $\pm$ symbol on the left hand side of the previous equation is replaced by a minus
sign on the right.

Using the fact that $Q_f$ has a well defined sign and that $x$ is Lipschitz and
thus $-k \leq \frac{dx}{dt} \leq k$, equation~(\ref{eq:dt_epsilon}) allows us to bound the
derivative of $\epsilon$. Gathering the upper and lower bounds in a single differential
inequation we have:

\begin{equation} \label{ineq:main}
-\epsilon \diff{\cdot}\alpha\diff{\cdot} Q_f(x, \epsilon) - k \diff{\cdot}\frac{df}{dx}
\leq \frac{d\epsilon}{dt} \leq 
-\epsilon \diff{\cdot}\alpha \diff{\cdot}Q_f(x, \epsilon) + k \diff{\cdot}\frac{df}{dx}.
\end{equation}

This inequality is our central result as it allows to determine the sign of
$\frac{d \epsilon}{dt}$ and thus the flow of the error. If the left hand side is positive,
$\epsilon$ is increasing and if the right hand side is negative it is decreasing. In the
third case, that is if the left and right
hand side have different signs, we are not able to determine the sign of $\frac{d
\epsilon}{dt}$ \diff{\sout{and do not control the error}}.
Thus, this inequality~(\ref{ineq:main}) separates the space in $3$ different
regions. The ones where the derivative is positive or negative for sure, and the one where
we cannot determine its sign. These three regions are delimited by the
two curves where the right and left hand side vanished that is:
\begin{equation} \label{eq:boundary}
   \epsilon \diff{\cdot}Q_f(x, \epsilon) = \pm \frac{k}{\alpha} \diff{\cdot}\frac{df}{dx}.
\end{equation}

A consequence of the flow of the differential inequation is that a system that start in
the inner region where the sign is unknown will always stays in this region whatever the
variation of the input, hence giving us an upper and lower bound on the output. The
important parameter controlling this region is the ratio between the Lipschitz coefficient
of the input $k$ and the kinetic turn-over $\alpha$ \diff{since the error results from a kind
of race between the target function controlled by $k$, and the output controlled by $\alpha$}. An interesting consequence is that
the bounding becomes tighter when the input function stabilizes and $k$ decreases. The
precise form of this region \diff{\sout{is}can be} directly deduced from the polynomial $P_f$.

This brings us to the notion of online control of the error\diff{sout{. If we are}, when searching} for a
particular precision $p \in \R^+$ for the output given a constraint on the input to be $k$
Lipschitz, \diff{\sout{that is if we want to ensure} ensuring} that $\forall t, |\epsilon(t)| < p$.
Looking on equation~(\ref{eq:boundary}), we have to bound the right hand side.
\diff{\sout{, but this is
not necessarily possible on an open interval as
a minimum to $Q_f$ and a maximum to $\left| \frac{df}{dx} \right|$
may be on the boundary, or even diverge. This is no more a problem if we restrict our
analysis to a closed region $I \in \D$ where the extrema are well defined. We thus have
the following result.}
For this, we have to find a minimum to $Q_f$ and a maximum to $\left| \frac{df}{dx}
\right|$. To prevent these quantities from diverging on the boundaries of their domain, let us
restrict ourself to a closed region $I \subset \D$.}

\begin{proposition}
\diff{
   \sout{
   For input functions that are $k$-Lipschitz and a system in a closed region $I \subset \D$,
   there exist a kinetic parameter $\alpha^\star(p,I)$:
   such that, if $\alpha \geq \alpha^\star$
   and $|\epsilon(t=0)| \leq p$, then $\forall t, |\epsilon(t)| \leq p$
   with:

   $ \text{definition of } \alpha^\star$
   }
}
\end{proposition}

\diff{
\begin{proposition}\label{prop:precision}
   For input functions that are $k$-Lipschitz and a system in a closed region $I \subset \D$,
   we can define the quantity:
\begin{equation}
	\alpha^\star(p, I) =  \frac{k}{p} \frac{\max_I{\sum_i \left| \frac{\partial f}{\partial
   x_i} \right|}}{\min_I{Q_f}}
\end{equation}
   ensuring that for $\alpha \geq \alpha^\star$, we have:
   $|\epsilon(t_1)| \leq p \Rightarrow \forall t > t_1, |\epsilon(t)| \leq p$.
	That is, as soon as the precision $p$ is reached, it will be conserved at all time.
\end{proposition}
}

Before showing on an example the derivation of the stable region. Let us have a last glance on
the equation~(\ref{eq:dt_epsilon}) to discuss the transient response if the system start
outside of the stable region defines by equation~(\ref{eq:boundary}). Assuming that the system
is nonetheless in the domain $\D$
and will not diverge or reach another branch of the algebraic curve. The transient regime
is important when the last term of equation~(\ref{ineq:main}) is negligible and thus, we will
consider the case where $k=0$, that is the inputs are fixed.
The main argument now is that we have suppose that the polynomial $P_f$ was of
multiplicity \diff{\sout{one}$1$} for the branch we want to stabilize, or said otherwise: $Q_f(X,\epsilon)$
is such that for all initial conditions $X$, $\epsilon = 0$ is not a root of $Q_f$. The
direct consequence of this remark is that, at least locally around the solution, the
convergence toward the stable region is always exponential with a characteristic time of
the order of $\frac{1}{\alpha \diff{\cdot}Q_f(X, 0)}$. A more rigorous proof of this assertion can be
found in~\cite{FKTLNR21ucnc}.

\vspace{1em}

To sum up the behavior of a stabilizing CRN when the input functions are $k$-Lipschitz.
We first have a transient response of exponential decay toward the desired target
function. And once the system enter into a region the shape of which is defined by $P_f$
and the width by the ratio $\frac{k}{\alpha}$, we know for sure that it will remains in
this bounded region.

This gives a mitigate result for the precision, on one side we first
have a fast increase in precision when the system is far from the target function because
the exponential decay means that we gain one bit of precision at each time step (up to a
multiplicative factor). But once we have reached the bounded region, as $\alpha^\star(p) \propto
\frac{1}{p}$, each supplementary bit of precision asks for a factor $2$ in the kinetic
turn-over which may represent a high energetic cost.

\begin{example}
	To clarify the general case above, we will \diff{first} develop our framework on the simplest case:
$f(x) = x^2$. Thus $P_f(x,y) = x^2 - y$, and we have:
\begin{align}
    \frac{dy}{dt} &= x^2 - y, \\
    \frac{df}{dx} &= 2\diff{\cdot}x, \\
    Q_f(x,\epsilon) &= 1.
\end{align}

	We can then rewrite inequation~(\ref{ineq:main}) for this particular case:
\begin{equation}
    -2\diff{\cdot}k\diff{\cdot}x - \epsilon \diff{\cdot}\alpha \leq \frac{d \epsilon}{dt} \leq -\epsilon \diff{\cdot}\alpha +2\diff{\cdot}k\diff{\cdot}x.
\end{equation}

	Or taking directly the boundaries given by equation~(\ref{eq:boundary}): $\epsilon = \pm
2\frac{k}{\alpha}x$. This gives a bounding interval for the output function:
	\begin{equation}
		x^2 - 2\frac{k}{\alpha}x \leq y(t) \leq x^2 + 2\frac{k}{\alpha}x .
	\end{equation}
	The idea is that while the exact behaviour of the output is unknown without adding more
	constraints on the input, we can at least ensure this permissive bounding on the
	system as presented in figure~\ref{fig:error_control}\textbf{A}.
\end{example}

\begin{figure}
	\centering
	\textbf{A.} \includegraphics[width=.45\textwidth]{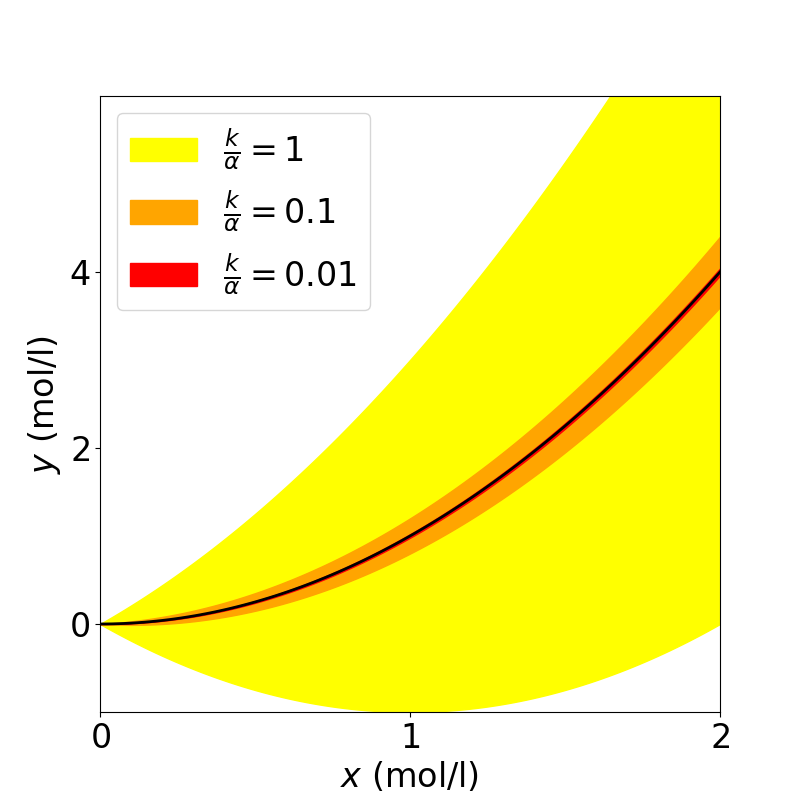}
	\textbf{B.} \includegraphics[width=.45\textwidth]{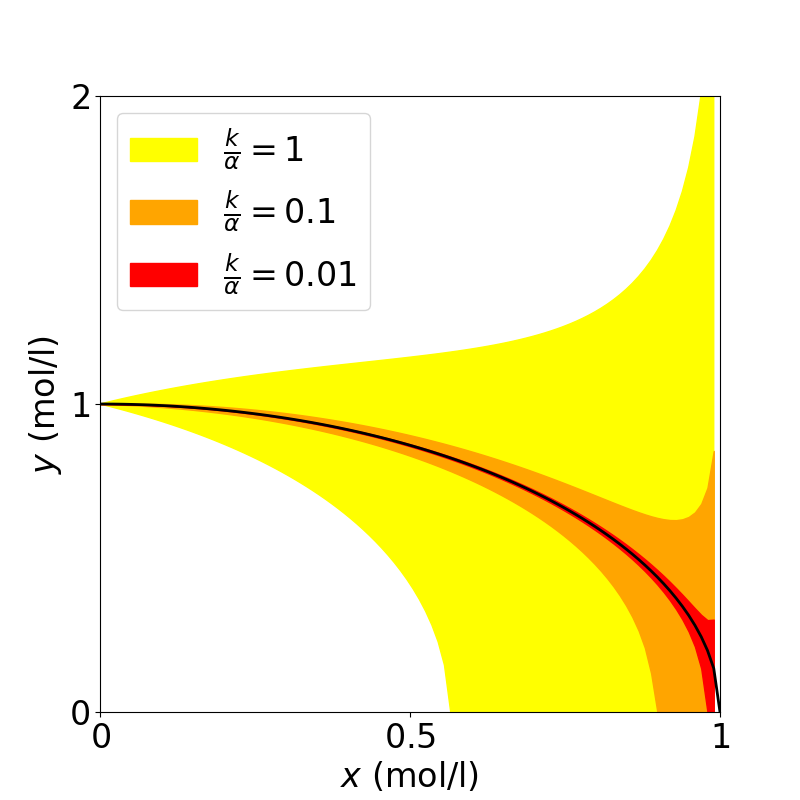}
	\caption{Bounding error regions for various values of $\frac{k}{\alpha}$.
	\textbf{A.} $f(x) = x^2$.
	\textbf{B.} $f(x) = \sqrt{1-x^2}$.
	}\label{fig:error_control}
\end{figure}

\begin{example}
	Now we will look at the error control for the CRN stabilizing the circle already
	presented in example~\ref{ex:circle}. We can easily derived the following
	relations for the main quantities:
	\begin{gather}
		\begin{aligned}
			f(x) &= \sqrt{1-x^2}, &\quad \frac{df}{dx} &= -\frac{x}{f(x)}, \\
			P_f(x,y) &= 1-x^2-y^2, &\quad Q_f(x,\epsilon) &= - \epsilon - 2 \diff{\cdot}f(x).
		\end{aligned}
	\end{gather}
	Let us make a pause on the roots of $Q_f$ ($\epsilon = -2\diff{\cdot} f(x)$) to note that they are
	all on the lower branch of the circle and thus outside of the domain $\D$. Hence
	validating the fact that we have an exponential convergence for all points inside the
	domain.

	Now inserting the computed quantities into the general case
	equation~(\ref{eq:boundary}), we have:
	\begin{equation}
		\epsilon^2 + 2 \diff{\cdot}\epsilon \diff{\cdot}f(x) \pm \frac{k}{\alpha} \frac{x}{f(x)} = 0
	\end{equation}
	Solving this quadratic equation, we obtain the upper and lower error bounds which gives
	us a bounding for the output when the input is $k$-Lipschitz:
	\begin{equation}
		f(x) \sqrt{1 - \frac{k}{\alpha} \frac{x}{f^3(x)}} \leq
		y(t) \leq f(x) \sqrt{1 + \frac{k}{\alpha} \frac{x}{f^3(x)}}.
	\end{equation}

	As can be seen on figure~\ref{fig:error_control}\textbf{B.}, a peculiar behaviour of
	this inequality is that there exists a region near $x=1$ where the lower bound is
	undefined. Intuitively, this is linked to the possibility, when close from the
   monodromy point $(x=1, y=0)$ of crossing the lower branch of the circle and from there
   the flow will drive the system to $y \rightarrow - \infty$ as it is no more on the
   $\D$ domain (see fig.~\ref{fig:circle}).
\end{example}

While these regions may appear quite large, they corresponds to a powerful notion of error
control in the sense that they ensure a given precision for any \diff{k-}Lipschitz input functions,
which is a small constraint. \diff{Moreover, }in practice, the error observed for typical input functions
like offset sine waves are far smaller than the bounds given by our analysis.

\subsection{Error control with intermediate species}

Our main tool to transform the ODE so that they are amenable to a CRN implementation is
through the introduction of new variables. This rises the question of how the final error
is affected by this kind of transformation. 
We will not provide a full analysis of this phenomena in the general case but
will build an intuition upon two basic examples corresponding to the two main cases where
we have to introduce new variables: the quadratization when dealing with monomial of
degree higher than $2$ and the polynomialization when the expression provided by the user
is not directly a polynomial.

\begin{example}
We will start with the simplest case of quadratization taking $y = x^4$ as
our direct case and $y = z^2 = x^4$ as the chained case.

For the direct case we just repeat the computation of the previous subsections and
immediately obtain:
\begin{equation}
    y = x^4 \pm 4 \frac{k}{\alpha} x^3.
\end{equation}

For the quadratized case, we rely on the computation of the square case to obtain:
\begin{equation}
	\begin{aligned}
		y &= z^2 \pm 2 \frac{k}{\alpha} z, \\
		  &= x^4 \pm 4 \frac{k}{\alpha} x^3 \pm 4 \left( \frac{k}{\alpha} \right)^2 x^2
					\pm 2 \frac{k}{\alpha} x^2 \pm 4 \left( \frac{k}{\alpha} \right)^2 x.
	\end{aligned}
\end{equation}
and we can see that the first error term is the same as the one of the direct case why all
the other ones are supplementary errors introduced by the quadratization, meaning that the
quadratization is a net source of errors. This was expected as it introduces delay along
the computation so that the output is less reactive to a variation of the input.
The supplementary error introduced by the quadratization is especially important when
	$\frac{k}{\alpha}>1$ and $x \ll 1$ increasing the importance of the ratio between the
	regularity of the input and the kinetic turnover for quadratized systems.
\end{example}

\begin{figure}
	\centering
	\includegraphics[width=.6\textwidth]{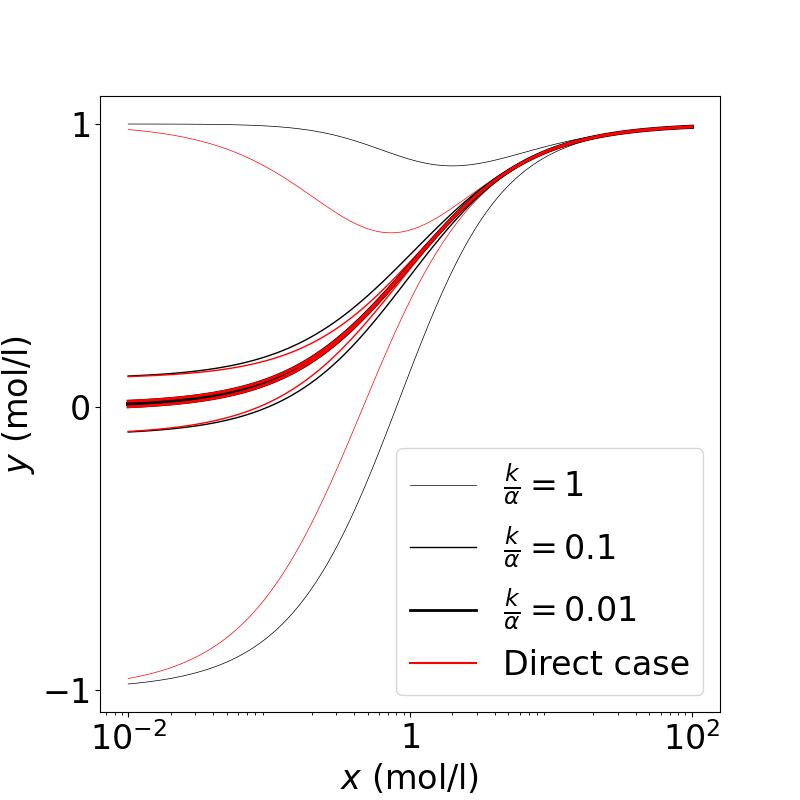}
	\caption{Bounding error regions for the hill function using the direct computation (in
	red) and the intermediate species of the inverted hill function (in
	black). The red region is always tighter than the black one indicating the loss of
	precision introduced by the new species. The difference between the two methods is
	especially important around the $x=1$ region.
	}\label{fig:error_hill}
\end{figure}

\begin{example}
	A more involved case is the stabilization of the Hill function $f(x) = \frac{x}{1+x}$.
	While this case may seem somehow artificial as the direct implementation is possible
	and shows a better bounding on the error. The CRN with one variable added is actually
	simpler to implement when going to higher order hill functions and may effectively be
	produced by our algorithm depending on the way the expression is provided to the
	compiler (see example~\ref{ex:hill}).

	The simplest way to stabilize a hill function is thus through a direct implementation:
	\begin{equation}
		\frac{dy}{dt} = x - y - x\diff{\cdot}y.
	\end{equation}
	computing the error in this case is easy and we obtain:
	\begin{equation}
		|\epsilon| \leq \frac{k}{\alpha} \frac{1}{(1+x)^3}.
	\end{equation}
	But another path is to introduce the inverse of the hill function $z = \frac{1}{1+x}$
	and use it to construct the desired output:
	\begin{gather}
		\begin{aligned}
			\frac{dz}{dt} &= 1 - z - z\diff{\cdot}x \\ 
			\frac{dy}{dt} &= z\diff{\cdot}x - y
		\end{aligned}
	\end{gather}
	In this case, the error on $z$ is exactly the same as the computation of the direct
	path but this error is then passed on $y$ giving us a final error on the output:
	\begin{equation}
		|\epsilon| \leq \frac{k}{\alpha} \left(\frac{x}{(1+x)^3}+\frac{1}{(1+x)^2}\right),
	\end{equation}
	which is, once again, strictly worse than the direct case as can be seen on
	figure~\ref{fig:error_hill}. The error in the direct and indirect path are similar when
	the input is far from $1$ but in the region near the threshold of the hill function,
	the error of the indirect case may be far larger than the direct case. Moreover, it is
	impossible to increase $\alpha$ to mitigate this effect as both cases scale similarly
	with respect to the kinetic turnover.
\end{example}

To conclude this section, using intermediate species to perform the computation is always
detrimental in terms of error control.
The two main sources of error being the amplification of intermediate errors and the
computational delay making the output less sensitive to fast variations of the input
functions.

While the first kind of errors may be handle through an increase of the kinetic turnover,
which is important because of there prevalence in the quadratization scheme. The second
kind cannot be reduced except by a better choice during the reduction of the expression to
be compiled.

For now, our polynomialization algorithm does not incorporate any form of choice based on
the scaling of the error but it could be interesting to see, when several
polynomialization or quadratization are possible, if we could determine the ones producing
the tighter region for the final output.

\section{Conclusion and perspectives}

We have introduced a notion of \diff{Absolute Functional Robustness 
  for CRNs that compute a real function} online\sout{analog chemical computation of real functions in which}\diff{, in the sense that} the
concentration of one output species continuously stabilizes to the result of some function
of the input species concentrations, whatever changes are applied to
the molecular concentrations within the domain of definition of the function. We
have shown that the set of real functions that can be continuously stabilized by a CRN \diff{with mass action law kinetics}
is precisely the set of real algebraic functions.

Furthermore, by restricting the changes of the inputs to continuous changes with a rate bounded by some constant,
i.e.~by restricting to \diff{environment inputs driven by} $k$-Lipschitz functions,
we were able to bound the error on the output by an analytic expression which can in principle be used  to control the error.

\diff{While o}ur main theorem  focuses on CRNs at steady states, \diff{some} important aspects of signal
processing are linked to the timings of the signals.
Our CRN stabilization framework relies on ratios between molecular production and degradation,
and the error analysis we have done relates the possible multiplication of both terms by some factor $\alpha$
to  the characteristic time \diff{\sout{$\tau$}$T_\alpha$} of relaxation, with \diff{\sout{$\tau \propto \frac{1}{\alpha}$}$T_\alpha \propto \frac{1}{\alpha}$}.
It is thus worth remarking in this respect that high
value of \diff{\sout{$\tau$}$T_\alpha$} filter out the high frequency noise of the inputs, while small \diff{\sout{$\tau$}$T_\alpha$} values
result in a more accurate output, yet at the expense of a higher molecular turnover.

These theoretical results open a whole research avenue for both the understanding of the structure of
natural CRNs that allow cells to adapt to their environment, and for the design of
artificial CRNs to implement in chemistry high-level functions that have to be maintained in moving environments.
\diff{For these reasons, o}ur online computational framework \diff{seems} better suited 
than our earlier definition of functions computed by a CRN with fixed initial conditions~\cite{FLBP17cmsb},
to both the formal study of natural CRNs in the perspective of systems biology,
and the design of robust artificial CRNs in the perspective of synthetic biology.

The CRN compiler of algebraic functions we have implemented in Biocham according to these
principles, automatically generates an abstract CRN.
Implementing \diff{that abstract CRN with a concrete CRN using} with real enzymes, as done e.g.~in~\cite{CAFRM18msb} \diff{for the making of protocellular biosensors}, 
\diff{\sout{is not necessarily
possible for a given choice of generated abstract CRN while it could be possible for
another choice.}
may depend on a non trivial way of the different choices made during the steps of
polynomialization, dual-rail encoding and quadratization, rendering the final abstract CRN prone
to a concrete implementation, or not.}
Taking into account a catalog of concrete enzymatic reactions
earlier on in our compilation pipeline, i.e. \diff{\sout{in}during} the polynomialization, quadratization and
dual-rail encoding steps, is a particularly interesting challenge.
First, from the point of view of the potential applications in the biomedical and environment domains,
to guide search towards concrete economical solutions,
but also from the point of view of the computational complexity of the computer algebra problems that need to be solved in our compilation scheme,
e.g.~for guiding search in our quadratization problem shown NP-hard in~\cite{HFS20cmsb} and currently solved in practice using a MAXSAT solver.

\subsubsection*{Acknowledgments.} We are grateful to Amaury Pouly and Sylvain Soliman for preliminary discussions on this work,
to the reviewers for their very insigthful comments,
and to ANR-20-CE48-0002 Difference and Inria AEx GRAM grants for partial support.

\bibliographystyle{plain}
\bibliography{contraintes}

\begin{thebibliography}{10}

\bibitem{BEHL09al}
H.~J. Buisman, H.~M.~M. ten Eikelder, P.~A.~J. Hilbers, and A.~M.~L. Liekens.
\newblock Computing algebraic functions with biochemical reaction networks.
\newblock {\em Artificial Life}, 15(1):5--19, 2009.

\bibitem{CFS06bi}
Laurence Calzone, Fran{\c{c}}ois Fages, and Sylvain Soliman.
\newblock {BIOCHAM}: An environment for modeling biological systems and
  formalizing experimental knowledge.
\newblock {\em Bioinformatics}, 22(14):1805--1807, 2006.

\bibitem{CTT20nc}
Luca Cardelli, Mirco Tribastone, and Max Tschaikowski.
\newblock From electric circuits to chemical networks.
\newblock {\em Natural Computing}, 19, 2020.

\bibitem{CPSW05ejde}
David~C. Carothers, G.~Edgar Parker, James~S. Sochacki, and Paul~G. Warne.
\newblock Some properties of solutions to polynomial systems of differential
  equations.
\newblock {\em Electronic Journal of Differential Equations}, 2005(40):1--17,
  2005.

\bibitem{CLN13issb}
Vijayalakshmi Chelliah, Camille Laibe, and Nicolas Nov{\`e}re.
\newblock Biomodels database: A repository of mathematical models of biological
  processes.
\newblock In Maria~Victoria Schneider, editor, {\em In Silico Systems Biology},
  volume 1021 of {\em Methods in Molecular Biology}, pages 189--199. Humana
  Press, 2013.

\bibitem{CDS12nc}
Ho-Lin Chen, David Doty, and David Soloveichik.
\newblock Deterministic function computation with chemical reaction networks.
\newblock {\em Natural computing}, 7433:25--42, 2012.

\bibitem{CSWB09ab}
Matthew Cook, David Soloveichik, Erik Winfree, and Jehoshua Bruck.
\newblock Programmability of chemical reaction networks.
\newblock In Anne Condon, David Harel, Joost~N. Kok, Arto Salomaa, and Erik
  Winfree, editors, {\em Algorithmic Bioprocesses}, pages 543--584. Springer
  Berlin Heidelberg, Berlin, Heidelberg, 2009.

\bibitem{CAFRM18msb}
Alexis Courbet, Patrick Amar, Fran{\c{c}}ois Fages, Eric Renard, and Franck
  Molina.
\newblock Computer-aided biochemical programming of synthetic microreactors as
  diagnostic devices.
\newblock {\em Molecular Systems Biology}, 14(4), 2018.

\bibitem{CF06siamjam}
Gheorghe Craciun and Martin Feinberg.
\newblock Multiple equilibria in complex chemical reaction networks: {II}. the
  species-reaction graph.
\newblock {\em SIAM Journal on Applied Mathematics}, 66(4):1321--1338, 2006.

\bibitem{DWGLERBW14nar}
X.~Duportet, L.~Wroblewska, P.~Guye, Y.~Li, J.~Eyquem, J.~Rieders, G.~Batt, and
  R.~Weiss.
\newblock A platform for rapid prototyping of synthetic gene networks in
  mammalian cells.
\newblock {\em Nucleic Acids Research}, 42(21), 2014.

\bibitem{ET89book}
P{\'e}ter {\'E}rdi and J{\'a}nos T{\'o}th.
\newblock {\em Mathematical Models of Chemical Reactions: Theory and
  Applications of Deterministic and Stochastic Models}.
\newblock Nonlinear science : theory and applications. Manchester University
  Press, 1989.

\bibitem{FLBP17cmsb}
Fran\c{c}ois Fages, Guillaume Le~Guludec, Olivier Bournez, and Amaury Pouly.
\newblock {Strong Turing Completeness of Continuous Chemical Reaction Networks
  and Compilation of Mixed Analog-Digital Programs}.
\newblock In {\em {CMSB'17}: Proceedings of the fiveteen international
  conference on Computational Methods in Systems Biology}, volume 10545 of {\em
  Lecture Notes in Computer Science}, pages 108--127. Springer-Verlag,
  September 2017.

\bibitem{FGS15tcs}
Fran{\c{c}}ois Fages, Steven Gay, and Sylvain Soliman.
\newblock Inferring reaction systems from ordinary differential equations.
\newblock {\em Theoretical Computer Science}, 599:64--78, September 2015.

\bibitem{FS08tcs}
Fran{\c{c}}ois Fages and Sylvain Soliman.
\newblock Abstract interpretation and types for systems biology.
\newblock {\em Theoretical Computer Science}, 403(1):52--70, 2008.

\bibitem{Feinberg77crt}
Martin Feinberg.
\newblock Mathematical aspects of mass action kinetics.
\newblock In L.~Lapidus and N.~R. Amundson, editors, {\em Chemical Reactor
  Theory: A Review}, chapter~1, pages 1--78. Prentice-Hall, 1977.

\bibitem{FKTLNR21ucnc}
Willem Fletcher, Titus~H Klinge, James~I Lathrop, Dawn~A Nye, and Matthew
  Rayman.
\newblock Robust real-time computing with chemical reaction networks.
\newblock In {\em Unconventional Computation and Natural Computation: 19th
  International Conference, UCNC 2021, Espoo, Finland, October 18--22, 2021,
  Proceedings 19}, pages 35--50. Springer, 2021.

\bibitem{HT79cmsjb}
V.~H{\'a}rs and J.~T{\'o}th.
\newblock On the inverse problem of reaction kinetics.
\newblock In M.~Farkas, editor, {\em Colloquia Mathematica Societatis J{\'a}nos
  Bolyai}, volume~30 of {\em Qualitative Theory of Differential Equations},
  pages 363--379, 1979.

\bibitem{HF22cmsb}
Mathieu Hemery and Fran{\c{c}}ois Fages.
\newblock Algebraic biochemistry: a framework for analog online computation in
  cells.
\newblock In {\em {CMSB'22}: Proceedings of the twentieth international
  conference on Computational Methods in Systems Biology}, volume 13447 of {\em
  Lecture Notes in Computer Science}. Springer-Verlag, September 2022.

\bibitem{HFS20cmsb}
Mathieu Hemery, Fran{\c{c}}ois Fages, and Sylvain Soliman.
\newblock On the complexity of quadratization for polynomial differential
  equations.
\newblock In {\em {CMSB'20}: Proceedings of the eighteenth international
  conference on Computational Methods in Systems Biology}, Lecture Notes in
  Computer Science. Springer-Verlag, September 2020.

\bibitem{HFS21cmsb}
Mathieu Hemery, Fran{\c{c}}ois Fages, and Sylvain Soliman.
\newblock Compiling elementary mathematical functions into finite chemical
  reaction networks via a polynomialization algorithm for {ODE}s.
\newblock In {\em {CMSB'21}: Proceedings of the nineteenth international
  conference on Computational Methods in Systems Biology}, volume 12881 of {\em
  Lecture Notes in Computer Science}. Springer-Verlag, September 2021.

\bibitem{HF19jpcb}
Mathieu Hemery and Paul Fran{\c{c}}ois.
\newblock In silico evolution of biochemical log-response.
\newblock {\em The Journal of Physical Chemistry B}, 2019.

\bibitem{HF96pnas}
Chi-Ying Huang and James~E. Ferrell.
\newblock Ultrasensitivity in the mitogen-activated protein kinase cascade.
\newblock {\em PNAS}, 93(19):10078--10083, September 1996.

\bibitem{Hucka03bi}
Michael Hucka et~al.
\newblock The systems biology markup language ({SBML}): A medium for
  representation and exchange of biochemical network models.
\newblock {\em Bioinformatics}, 19(4):524--531, 2003.

\bibitem{sbml20msb}
Sarah~M. Keating, Dagmar Waltemath, Matthias K{\"o}nig, Fengkai Zhang, Andreas
  Dr{\"a}ger, Claudine Chaouiya, Frank~T. Bergmann, Andrew Finney, Colin~S.
  Gillespie, Tom{\'a}{s} Helikar, Stefan Hoops, Rahuman~S. Malik-Sheriff,
  Stuart~L. Moodie, Ion~I. Moraru, Chris~J. Myers, Aur{\'e}lien Naldi, Brett~G.
  Olivier, Sven Sahle, James~C. Schaff, Lucian~P. Smith, Maciej~J. Swat, Denis
  Thieffry, Leandro Watanabe, Darren~J. Wilkinson, Michael~L. Blinov, Kimberly
  Begley, James~R. Faeder, Harold~F. G{\'o}mez, Thomas~M. Hamm, Yuichiro
  Inagaki, Wolfram Liebermeister, Allyson~L. Lister, Daniel Lucio, Eric
  Mjolsness, Carole~J. Proctor, Karthik Raman, Nicolas Rodriguez, Clifford~A.
  Shaffer, Bruce~E. Shapiro, Joerg Stelling, Neil Swainston, Naoki Tanimura,
  John Wagner, Martin Meier-Schellersheim, Herbert~M. Sauro, Bernhard Palsson,
  Hamid Bolouri, Hiroaki Kitano, Akira Funahashi, Henning Hermjakob, John~C.
  Doyle, Michael Hucka, Richard~R. Adams, Nicholas~A. Allen, Bastian~R.
  Angermann, Marco Antoniotti, Gary~D. Bader, Jan {C}erven{\'y}, M{\'e}lanie
  Courtot, Chris~D. Cox, Piero~Dalle Pezze, Emek Demir, William~S. Denney,
  Harish Dharuri, Julien Dorier, Dirk Drasdo, Ali Ebrahim, Johannes Eichner,
  Johan Elf, Lukas Endler, Chris~T. Evelo, Christoph Flamm, Ronan~MT Fleming,
  Martina Fr{\"o}hlich, Mihai Glont, Emanuel Gon\c{c}alves, Martin Golebiewski,
  Hovakim Grabski, Alex Gutteridge, Damon Hachmeister, Leonard~A. Harris,
  Benjamin~D. Heavner, Ron Henkel, William~S. Hlavacek, Bin Hu, Daniel~R.
  Hyduke, Hidde Jong, Nick Juty, Peter~D. Karp, Jonathan~R. Karr, Douglas~B.
  Kell, Roland Keller, Ilya Kiselev, Steffen Klamt, Edda Klipp, Christian
  Kn{\"u}pfer, Fedor Kolpakov, Falko Krause, Martina Kutmon, Camille Laibe,
  Conor Lawless, Lu~Li, Leslie~M. Loew, Rainer Machne, Yukiko Matsuoka, Pedro
  Mendes, Huaiyu Mi, Florian Mittag, Pedro~T. Monteiro, Kedar~Nath Natarajan,
  Poul~MF Nielsen, Tramy Nguyen, Alida Palmisano, Jean-Baptiste Pettit, Thomas
  Pfau, Robert~D. Phair, Tomas Radivoyevitch, Johann~M. Rohwer, Oliver~A.
  Ruebenacker, Julio Saez-Rodriguez, Martin Scharm, Henning Schmidt, Falk
  Schreiber, Michael Schubert, Roman Schulte, Stuart~C. Sealfon, Kieran
  Smallbone, Sylvain Soliman, Melanie~I. Stefan, Devin~P. Sullivan, Koichi
  Takahashi, Bas Teusink, David Tolnay, Ibrahim Vazirabad, Axel Kamp, Ulrike
  Wittig, Clemens Wrzodek, Finja Wrzodek, Ioannis Xenarios, Anna Zhukova, and
  Jeremy Zucker.
\newblock {SBML} level 3: an extensible format for the exchange and reuse of
  biological models.
\newblock {\em Molecular Systems Biology}, 16(8):1--21, August 2020.

\bibitem{OK11iet}
K.~Oishi and E.~Klavins.
\newblock Biomolecular implementation of linear i/o systems.
\newblock {\em IET Systems Biology}, 5(4):252--260, 2011.

\bibitem{QSW11dna}
Lulu Qian, David Soloveichik, and Erik Winfree.
\newblock Efficient turing-universal computation with {DNA} polymers.
\newblock In {\em Proc. {DNA} Computing and Molecular Programming}, volume 6518
  of {\em LNCS}, pages 123--140. Springer-Verlag, 2011.

\bibitem{Segel84book}
Lee~A. Segel.
\newblock {\em Modeling dynamic phenomena in molecular and cellular biology}.
\newblock Cambridge University Press, 1984.

\bibitem{SF10science}
Guy Shinar and Martin Feinberg.
\newblock Structural sources of robustness in biochemical reaction networks.
\newblock {\em Science}, 327(5971):1389--1391, 2010.

\bibitem{VSK18dna}
Marko Vasic, David Soloveichik, and Sarfraz Khurshid.
\newblock {CRN++}: Molecular programming language.
\newblock In {\em Proc. {DNA} Computing and Molecular Programming}, volume
  11145 of {\em LNCS}, pages 1--18. Springer-Verlag, 2018.

\end{thebibliography}

\end{document}